\documentclass[journal,12pt, onecolumn]{IEEEtran}

\usepackage{cite}
\usepackage{amsmath,amssymb,amsfonts}
\usepackage{algorithmic}
\usepackage{graphicx}
\usepackage{textcomp}
\usepackage{xcolor}
\usepackage{subfigure}
\usepackage{diagbox}
\usepackage{amsthm}
\usepackage{mathrsfs}

\usepackage{setspace}

\usepackage{pdfpages}
\usepackage{colortbl}
\definecolor{mygray}{gray}{0.85}
\usepackage{array}
\usepackage{caption}
\usepackage[linesnumbered,boxed,ruled,commentsnumbered]{algorithm2e}

 \usepackage{eqparbox}

\usepackage{amsmath,bm}
\usepackage{cite}
\usepackage{amsmath,amssymb,amsfonts}
\usepackage{algorithmic}
\usepackage{amsthm}
\usepackage{graphicx}
\usepackage{textcomp}
\usepackage{xcolor}

\begin{document}
\begin{spacing}{1.5}
%

\title{
\huge
Offloading Optimization in Edge Computing for Deep Learning Enabled Target Tracking by Internet-of-UAVs
}

\author{Bo~Yang,~\IEEEmembership{Member,~IEEE}, Xuelin~Cao, 
  Chau Yuen,~\IEEEmembership{Senior Member,~IEEE}, and~Lijun Qian,~\IEEEmembership{Senior Member,~IEEE}
\thanks{Bo Yang, Xuelin Cao and Chau Yuen are with the Engineering Product Development Pillar, Singapore University of Technology and Design, Singapore, 487372 (e-mail: bo$\_$yang, xuelin$\_$cao, yuenchau@sutd.edu.sg).}
\thanks{Lijun Qian is with the Department of Electrical and Computer Engineering and CREDIT Center, Prairie View A$\&$M University, Texas A$\&$M University System, Prairie View, TX 77446, USA (e-mail: liqian@pvamu.edu)}

}

\maketitle

\begin{abstract}
The empowering unmanned aerial vehicles (UAVs) have been extensively used in providing intelligence such as target tracking. In our field experiments, a pre-trained convolutional neural network (CNN) is deployed at the UAV to identify a target (a vehicle) from the captured video frames and enable the UAV to keep tracking. However, this kind of visual target tracking demands a lot of computational resources due to the desired high inference accuracy and stringent delay requirement. This motivates us to consider offloading this type of deep learning (DL) tasks to a mobile edge computing (MEC) server due to limited computational resource and energy budget of the UAV, and further improve the inference accuracy. Specifically,  we propose a novel hierarchical DL tasks distribution framework, where the UAV is embedded with lower layers of the pre-trained CNN model, while the MEC server with rich computing resources will handle the higher layers of the CNN model. An optimization problem is formulated to minimize the weighted-sum cost including the tracking delay and energy consumption introduced by communication and computing of the UAVs, while taking into account the quality of data (e.g., video frames) input to the DL model and the inference errors. Analytical results are obtained and insights are provided to understand the tradeoff between the weighted-sum cost and inference error rate in the proposed framework. Numerical results demonstrate the effectiveness of the proposed offloading framework.
\end{abstract}


\IEEEpeerreviewmaketitle

\section{Introduction}
\IEEEPARstart{D}{uring} the past decade, the unmanned aerial vehicles (UAVs), also commonly known as drones, have been extensively used in providing assorted appealing applications by leveraging UAVs for wireless communications for civilian, commercial and military services \cite{UAV_application}. Noticeably, visual tracking for UAV-captured target has recently gained much attention and has been extensively applied to anticipate crimes by remotely surveilling and tracking suspicious humans or vehicles at
places of interest~\cite{UAV_application02}. In visual target tracking scenarios, UAVs (or drones) collect video data from the cameras (e.g.,  high-resolution digital cameras) and try to detect and lock the target by processing the data frames in near-real-time, with the aid of digital signal processors (DSPs)~\cite{UAV_application03}. Since UAVs generally have severely limited power supply and low computing capability, they can hardly be able to complete the tasks requiring intensive computing by themselves, which impose great challenges on computing capability, low latency as well as requirement on the inference accuracy~\cite{Infocom_YB}.

\begin{figure}[t]
  \captionsetup{font={footnotesize }}
\centerline{ \includegraphics[width=4.9in, height=2.2in]{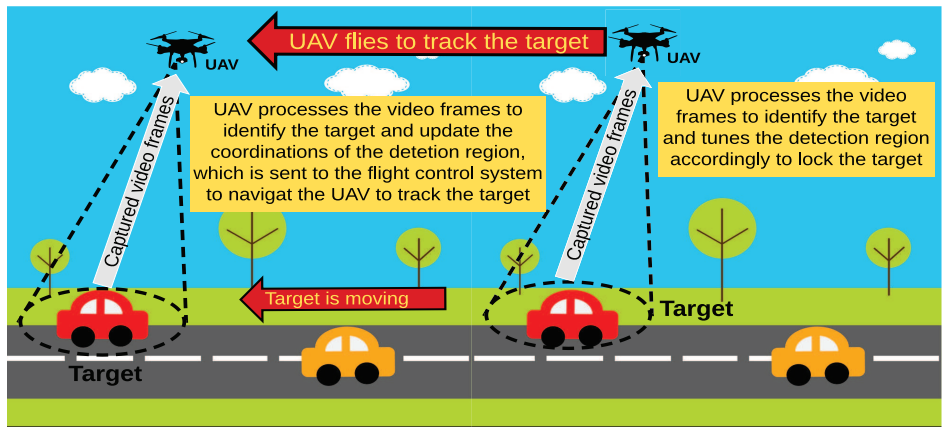}}
\caption{A typical visual target tracking scenario, where \textit{an embedded machine learning (ML) unit at UAV} identifies the target by processing the captured video frames and locks the target by tuning the detection region in near-real-time. Once the target moves fast, the ML unit not only needs to detect the target but also infer and update the coordinates of the detection region from the captured video frames. After that, the inference results are sent to the flight control system (FCS), which navigates the UAV to track the moving target.}
\label{scene}
\end{figure}

In the UAV-enabled aerial surveillance and visual target tracking scenarios, the UAVs usually need to detect and track the targets by processing the streamed video frames captured by the cameras mounted on UAVs in near-real-time, e.g., $20$ frames per second (fps).  A typical scenario of UAV-enabled aerial surveillance and tracking is highlighted in Fig.~\ref{scene}, where the camera is able to follow the target, and actively change its orientation and detection region to optimize for tracking performance based on the visual feedback results. Because those small-scale UAVs intrinsically have limited capabilities, processing a high volume of video streaming with high inference accuracy becomes infeasible. In such a situation, mobile/multi-access edge computing (MEC) is considered a promising technology to address these challenges by offloading the video images to a MEC server (MES) with much more computing capabilities (including computing resources and storage space)~\cite{MEC_survey}. Meanwhile, machine learning (ML) especially deep learning (DL), becomes increasingly popular in many computer vision-based applications~\cite{{UAV_AI01},{UAV_AI02}}. Since the traditional feature engineering is not always well suited for aerial tracking in complex environments, the deep neural networks (DNN)~\cite{survey-1},\cite{vtc}, especially the convolutional neural networks (CNNs) achieving state-of-the-art performance on image classification and recognition~\cite{cvpr}, will be applied to extract critical features from the captured video frames. Hence, it has been a general trend to construct an artificial intelligence (AI) based visual target tracking infrastructure for small-scale Internet-of-UAVs.

\subsection{Motivating Experiments} 
To demonstrate the necessity of tasks offloading for visual target tracking, we have designed and trained a modified CNN model for video processing and performed preliminary experiments~\cite{Wu}. In our experiments, we tested a UAV (DJI S1000 drone \cite{DJDrone}) tracking a specific vehicle, as shown in Fig.~\ref{expriment}(a). The UAV is equipped with an embedded GPU (NVIDIA Jetson TX2 \cite{TX2}), a camera (GoPro Hero 4) and some peripherals, as highlighted in Fig.~\ref{expriment}(b). Specifically, the camera captures the video in real-time ($24$ frames per second (fps)) and the video frames are input into the NVIDIA Jetson TX2 for processing. We build the ARM version of TensorFlow and install it on TX2. A pre-trained CNN model is running on TX2, which would process the received video frames and run the target tracking algorithm. The output is transformed into a control command and delivered from TX2 to the flight controller (Pixhawk 4 drone controller). We test the CNN model with $10,000$ video frames and measure the processing time. The experimental results show that TX2 can reach up to $7$ fps, which is far below the $24$ fps required for real-time video frames processing. As can be observed from the field experiments, how to support the neural networks to process the collected data in a timely fashion achieving a certain degree of inference accuracy becomes a challenging issue for UAV tracking due to its constrained carrying capacity. 
To address this concern, MEC provides a promising alternative solution by offloading the computation-intensive DL tasks (e.g., inferring the target by processing video frames within a tolerable delay) to the MES via the underlying radio access technology (RAT) such as WiFi or cellular networks (e.g., LTE) in 5G era~\cite{MEC_survey}.

\begin{figure}[t]
\centering
  \captionsetup{font={footnotesize }}
\subfigure[]{
\includegraphics[width=2.25in,height=1.65in]{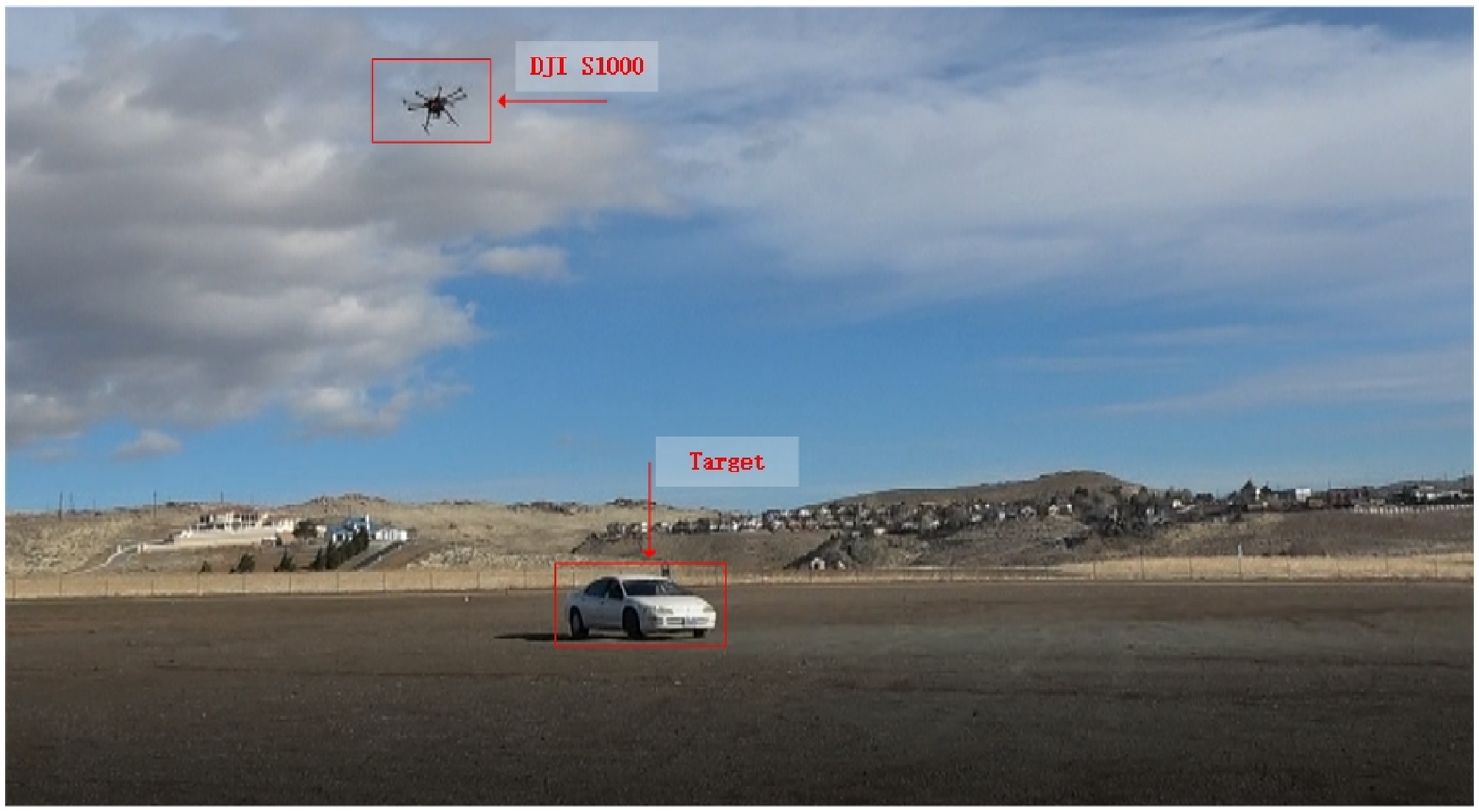}}
\hspace{-0.14in}
\subfigure[]{
\includegraphics[width=2.25in,height=1.65in]{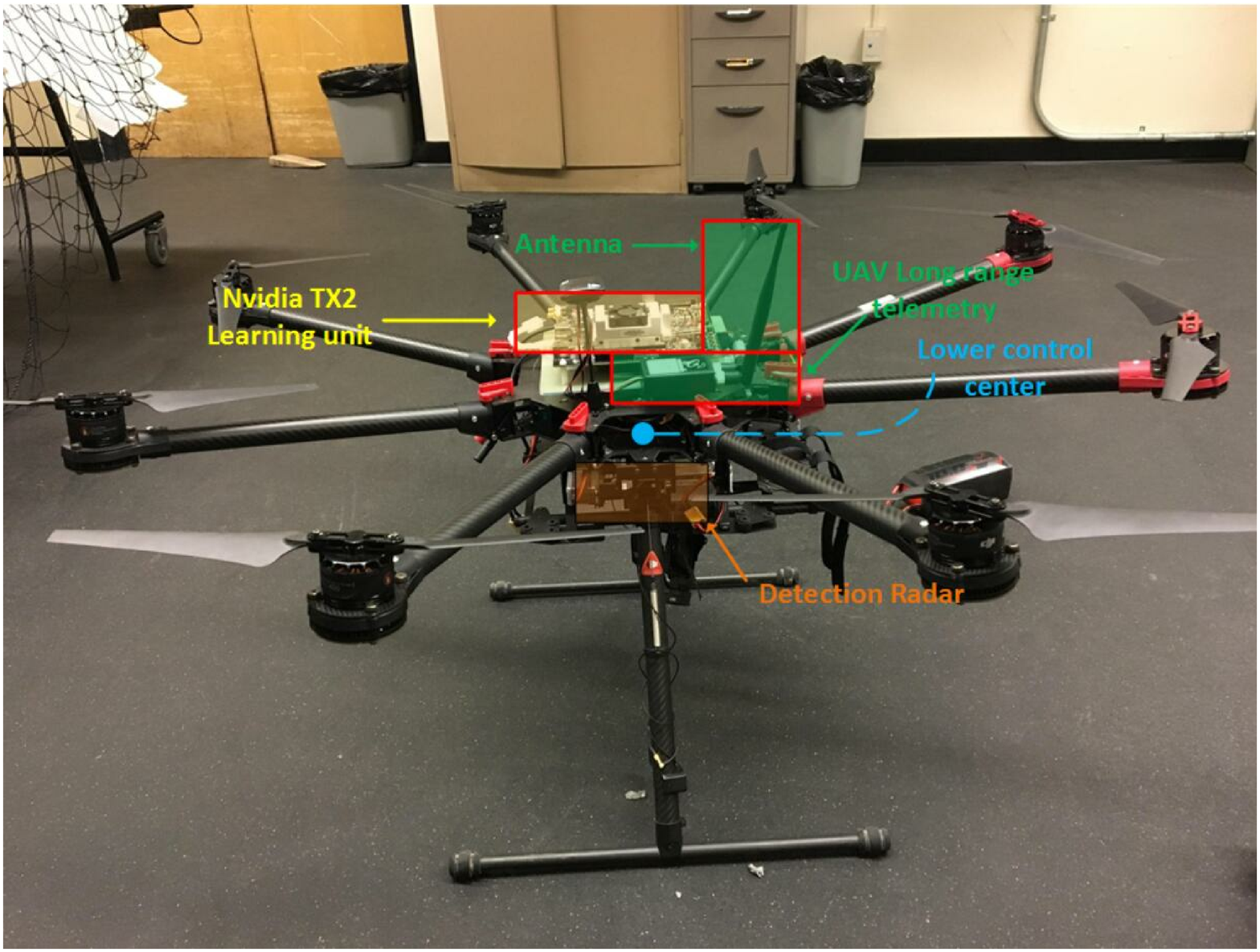}}
\caption{An envisioned visual target tracking test is shown in (a), where the testing UAV is tracking a white car, conducted in San Rafael park in Reno, Nevada, USA. The field testing UAV and its peripherals are detailed in (b).} 
\label{expriment}
\end{figure}


\subsection{Contributions and Paper Organization}
Motivated by the envisioned field testing, how to process DL tasks for the visual target tracking in a Internet-of-UAVs setup through computation offloading will be the focus of this work. In this paper, we propose a hierarchical machine learning tasks distribution (HMTD) framework, which aims to minimize the total weighted-sum cost of the UAVs with the inference error rate constraint by jointly considering the quality of data input to the DL model\footnote{This is hereafter abbreviated as ``data quality".}, computing capability at the UAV and the MES, and communications bandwidth. The main contributions of this paper are summarized as follows.

\begin{itemize}
\item \textbf{An HMTD framework:} We propose an HMTD framework for the deep learning based visual target tracking system to minimize the weighted-sum cost with the inference error rate constraint. In the proposed framework, the lower-level layers of the deep learning model are implemented at the UAV, while the higher-level layers are deployed at the MES. As a whole, the intermediate results generated by the UAV can be used directly or can be offloaded to the MES for further processing to improve the inference accuracy and decrease the total processing delay. After processed by the MES, the final results will be transferred back to the UAV to help in target tracking.

\item \textbf{Weighted-sum cost minimization problem and solution:} We formulate a weighted-sum cost minimization problem for both of the binary offloading and partial offloading schemes while taking into account the data quality, computing capability of the UAVs and the MES, and communications bandwidth. A closed-form optimal offloading probability and optimal offloading ratio is derived analytically for the binary offloading and partial offloading schemes, respectively. 

\item \textbf{Insights and results:} Some Insights are provided illustrating the effects of key parameters in the proposed offloading framework. This enables us to take advantage of the knowledge from ML research field for realistic visual target tracking scenarios. Numerical results are given to demonstrate the effectiveness of the proposed HMTD framework with the optimized offloading scheme.
\end{itemize}

\begin{table}[t]  
  \captionsetup{font={footnotesize}} 
\caption{ List of Key Notations } 
\small
\label{notations}  
\centering  
\begin{tabular}{|c | l|}  
\hline
\textbf{Notation} & \textbf{Description} \\
\hline 
$U_i$ & The $i$-th offloading UAV \\
\hline
$n$ & Total number of offloading UAVs \\
\hline
$\bf N$ & Offloading UAVs set \\
\hline
$l_L$ & Number of lower-level layers \\
\hline
$l_H$ & Number of higher-level layers \\
\hline
$K$ & Total number of deep learning layers \\
\hline
$\epsilon_L$ & Inference error rate given by UAVs \\
\hline
$\epsilon_H$ & Inference error rate given by the MES \\
\hline
$\epsilon_T^i$ & Inference error rate threshold of $U_i$ \\
\hline
$\eta$ & Percentage of data with bad quality \\
\hline
${\cal J}_i$ & The deep learning task of $U_i$\\
\hline
$s_i$ & Size of ${\cal J}_i$\\ 
\hline
$c_i$ & CPU cycles required to process  ${\cal J}_i$\\ 
 \hline
$\sigma_i$ & Maximum tolerable delay of ${\cal J}_i$\\ 
 \hline
$R_i$ & Achieved data rate between $U_i$ and the MES\\ 
 \hline
$f_l^i$ & CPU cycle frequency of $U_i$\\ \hline
$f_m$ & CPU cycle frequency of the MES \\
 \hline
$\mu_i$ & Offloading probability of $U_i$\\ \hline
$\mu_i^*$ & Optimal offloading probability of $U_i$\\ \hline
$\beta_i$ & Offloaded ratio of $U_i$ \\ \hline
$\gamma_i$ & Scale coefficient of data size output from $U_i$\\ \hline
$h_0$ & Channel gain at the reference distance of $1\ m$ \\
\hline
$\lambda_i$ & Distance between $U_i$ and the MES \\
\hline
$P_t^i$ & Transmission power of $U_i$\\ \hline
$P_I^i$ & Idle power of $U_i$\\ \hline
${\cal H}_i$ & Channel power gain for $U_i$ connecting with MES\\ \hline
$\rho_i$ & Failing penalty on the delay of $U_i$ \\
\hline
$\xi_i$ & Failing penalty on the energy consumption of $U_i$ \\
\hline
$\theta$ & UAV's preference on processing delay \\
\hline
$\tau_{l}^i$ &  Local execution delay of $U_i$\\ \hline
$\tau_{o}^i$ &  Execution delay of $U_i$ using offloading\\ \hline
$\varepsilon_{l}^i$ & Energy consumption processing ${\cal J}_i$ locally\\ \hline
$\varepsilon_{o}^i$ & Energy consumption processing ${\cal J}_i$ using offloading\\ \hline
${\cal O}_{l}^i$ & Weighted-sum cost processing ${\cal J}_i$ locally\\ \hline
${\cal O}_{o}^i$ & Weighted-sum cost processing ${\cal J}_i$ using offloading\\ \hline
${\cal O}_{i}^B$ & Weighted-sum cost of $U_i$ using binary offloading\\ \hline
${\cal O}_{total}^B$ & Total cost of all the UAVs using binary offloading\\ \hline
${\cal O}_{i}^P$ & Weighted-sum cost of $U_i$ using partial offloading\\ \hline
${\cal O}_{total}^P$ & Total cost of all the UAVs using partial offloading\\ \hline
\end{tabular}  
\end{table} 

 The remainder of this paper is organized as follows. Section~\ref{review} reviews the related works. In Section~\ref{system_model_section}, we present the system model and illustrate the proposed HMTD framework. In Section~\ref{binary_offloading_framework}, we present the optimization of the binary offloading framework, followed by the partial offloading optimization in Section~\ref{partial_offloading_framwork}. Some implementation issues are discussed and numerical results are presented in Section~\ref{results}. Finally, Section~\ref{conclusion} concludes this paper.

\textit{Notations:} As per the traditional notation, a bold letter indicates a vector and an upper case letter indicates a random variable or random parameter. $\max \{  \cdot \}$ and $\min \{  \cdot \}$ represent the maximum value and the minimum value, respectively.  Given a vector $\mathbf{x}$, then $\mathbf{x}^T$ denotes its transpose and $\left \| \mathbf{x} \right \|$ denotes its Euclidean norm. For ease of reference, Table \ref{notations} list some key notations.

\section{Literature Review} \label{review}

Recently, many pieces of literature concern the implementation of UAVs intending to improve the performance of the wireless communication system, where UAVs play the role of aerial surveillance and monitoring \cite{{UAV_Surveillance_survey1}, {UAV_Surveillance_survey2}, {UAV_tracking03}}, or as mobile relaying and ubiquitous coverage \cite{UAV_relay01, UAV_relay02, UAV_cover00,UAV_cover01}. 
To elaborate a little further, in \cite{UAV_Surveillance_survey1}, N. H. Motlagh \textit{et al.} introduced
the case of UAV-based crowd surveillance and developed a testbed using a built-in UAV along with a real-life LTE network. G. Ding \textit{et al.} \cite{UAV_Surveillance_survey2} developed an amateur drone surveillance system based on cognitive IoT, named Dragnet, tailoring the emerging cognitive internet of things framework for amateur drone surveillance. To fully explore the potential of multi-UAV sensor networks, in \cite{UAV_tracking03}, J. Gu \textit{et al.} proposed a new cooperative network platform and system architecture of multi-UAV surveillance. Moreover, Y. Zeng \textit{et al.} \cite{UAV_relay01} studied the throughput maximization problem in UAV relaying systems by optimizing the source/relay transmit power along with the relay trajectory. In \cite{UAV_relay02}, H. Wang  \textit{et al.} investigated the spectrum sharing planning problem for a full-duplex UAV relaying systems with underlaid D2D communications, where a mobile UAV employed as a full-duplex
relay assists the communication between separated nodes without a direct link. Furthermore, S.A.R. Naqvi \textit{et al.} \cite{UAV_cover00} presented a routing protocol for UAVs in disaster-resilient networks and presented a case study that incorporated UAVs in a wireless network equipped with both high- and low-power BSs. And in \cite{UAV_cover01}, M. Mozaffari \textit{et al.}  investigated the performance of a UAV that acts as a flying base station in an area in which users are engaged in the D2D communication.

With the development of MEC, there are assorted appealing applications by leveraging MEC techniques for wireless communications assisted by UAV. For example, a range of researchers in \cite{{UAV_MEC_1},{UAV_MEC_2},{UAV_MEC_3},{UAV_MEC_4}} proposed UAV-aided offloading systems, where the ground devices can be served by the flying UAVs. 
However, this kind of works is restricted by the capability of UAVs and can only be applicable for the large endurance UAVs endowed with computing capabilities to offer computation offloading services. 
With the rapid rise of small-scale commercially available UAVs which have several advantages in terms of cost, scalability, and survivability, the empowering smart UAVs with automated computer vision capabilities (e.g., object detection and tracking, etc.) is becoming a very promising research topic which has attracted the attention of industry and academia in the field \cite{UAV_Surveillance_survey1}.

In this context, with the aid of advantages of MEC, the tasks can be fully offload to MES or processed locally at the UAVs (denoted as binary offloading \cite{{Full-1},{Full-2},{YB_WCL}}), or only a portion of tasks is offloaded, which is indicated as partial offloading \cite{Partial-1}. However, binary offloading may neither satisfy the inference accuracy requirements due to constraints on the limited computing resources of the UAVs nor meet the demands of tolerable delay since additional wireless communication delay is introduced during the offloading. To meet the stringent delay requirement as well as achieve the inference accuracy of target tracking, it is desired to offload a proper portion of tasks to the MES. Furthermore, although DNN and MEC techniques are widely applied to enable delay-sensitive applications such as in industry settings~\cite{{UAV_DL1},{YC_TII}}, and vehicular networks~\cite{{XK_TVT},{XK_ITS}}, the inference error introduced by the DL model is rarely considered. However, in the practical DL applications, the inference errors will be affected by the quality of the data input to the neural networks. This is still an open issue for research of DL~\cite{error}.

\section{System Model and the Proposed Framework}\label{system_model_section}
As illustrated in Fig.~\ref{system_model}, a multi-UAV single-MES system is considered, where the UAVs are devoted to tracking a specific target, e.g., a vehicle or pedestrian. Suppose that there are total $n$ UAVs that may offload the tasks to the MES through LTE cellular network and the set of offloading UAVs is denoted as ${\bf N}=\{U_1,U_2,...,U_n\}$. 

\begin{figure}[t]
  \captionsetup{font={footnotesize }}
\centerline{ \includegraphics[width=3.75in, height=2.75in]{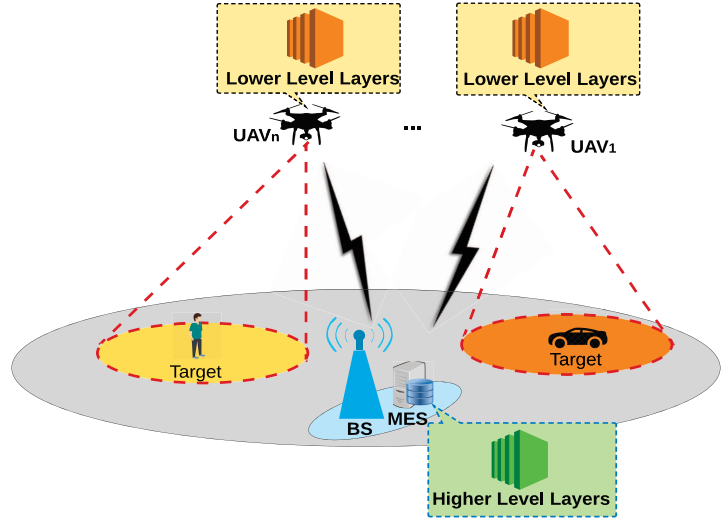}}
\caption{MEC based UAV tracking system model using the proposed hierarchical machine learning tasks distribution (HMTD) framework. Taking the CNN model as an example, the lower-level layers and higher-level layers consist of  the $Convolution + ReLU$ layers and $Pooling$ layers for the feature extraction, and the fully connected ($FC$) layers are deployed to generate the classification and regression results.}
\label{system_model}
\end{figure}

\subsection{The Proposed HMTD Framework}
In the visual target tracking system, we aim to optimize the system performance by designing an HMTD framework, where a deep learning model (e.g., the CNN model) is first pre-trained offline and then the trained model is further divided into two parts: {lower-level layers} and {higher-level layers}. The lower-level layers are deployed at the deep learning unit of the UAVs (e.g., NVIDIA Jetson TX2 unit in DJI S1000 done) and the higher-level layers are implemented at the MES co-located with an Base Station (BS).
The tasks are first processed locally with the lower-level layers saving the wireless bandwidth, whilst some of the intermediate data can be further offloaded to the MES with higher-level layers improving the inference accuracy.
To make the proposed HMTD framework easier to follow, some essential concepts are detailed as follows.

\theoremstyle{Definition} 

\newtheorem{definition}{\textit{Definition}}


\begin{definition}
\label{D1}
\textit{\textbf{Deep Learning Layers:} }For a pre-trained deep learning model (e.g., CNN in this paper), we define the layers of the DL model near input data as \textbf{lower-level layers} while the layers near output data is considered as \textbf{higher-level layers}. \end{definition}

\begin{figure*}[t]
  \captionsetup{font={footnotesize }}
\centerline{ \includegraphics[width=7.25in, height=3.55in]{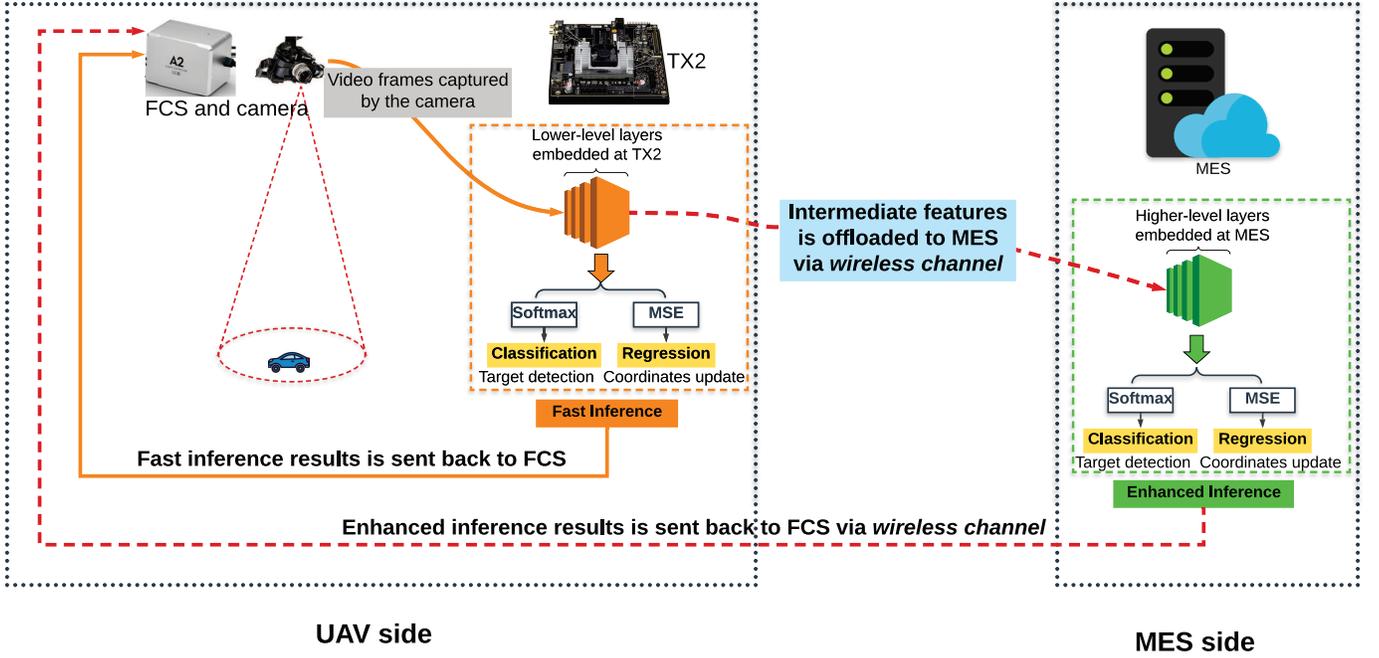}}
\caption{The proposed HMTD framework. In this framework, two inference modes (fast inference and enhanced inference) can be selected, where the fast inference is given by the lower-level layers embedded at the TX2 of the UAV and the enhanced inference is given by the higher-level layers deployed at the MES. The intermediate features data is offloaded to the MES and the inferring results can be transmitted back to the FCS of the UAV via the wireless channel to help the UAV to track the target.}
\label{NN_framework}
\end{figure*}

 As shown in Fig.~\ref{system_model}, taking the CNN model as an example, both of the lower-level layers and higher-level layers consist of $Convolution + ReLU$  layers and $Pooling$ layers for the feature extraction. Due to the constrained computation resources available in the UAV, it is reasonably assumed that lower-level layers are embedded at the UAV, which, however, makes it difficult to achieve inference with relatively high accuracy \cite{Harvard}, especially when the captured images are with low quality. The part with higher-level layers is deployed at the MES with more powerful computation resources. In this framework, the collected video frames are first fed into the lower-level layers and processed by the UAV. To further improve the inference accuracy, the intermediate features output from the UAV could be further offloaded to the MES with higher-level layers. 

\begin{definition}
\textbf{Inference Modes:} The inference given by the lower-level layers is called \textbf{fast inference}, and the inference performed by the higher-level layers is denoted as \textbf{enhanced inference}.
\end{definition}

In general, the enhanced inference outperforms the fast inference in terms of inference accuracy at the expense of introducing additional wireless transmission delay. In the visual target tracking system, there exist two branches for the inference: \textbf{1) Target detection}, which can be considered as a \textit{multi-class classification problem}, and 2) \textbf{Inferring the coordinates of the detection region}, which can be considered as a \textit{regression problem}. In the proposed HMTD framework, the multi-task learning can be adopted to optimize the loss functions of the two branches together~\cite{ICC_YB}, as illustrated in Fig.~\ref{NN_framework}.

\begin{definition}
\textbf{Offloading Strategy:} Two kinds of offloading modes are considered for each DL task (${\cal J}_i$) of $U_i$, i.e., \textbf{binary offloading} and \textbf{partial offloading}. In the binary offloading strategy,  the inference results are obtained either using the fast inference or using the enhanced inference. In the partial offloading strategy, the inference results could be obtained by both of the two inference modes.
\end{definition}

\theoremstyle{Remark} 
\newtheorem{remark}{{\textit{Remark}}}

\begin{definition}
\label{D3} 
 \textit{\textbf{Inference Error rate:}} Given a trained DL model (denoted as ${\cal D}$) and input data with a certain quality\footnote{In this paper, the input data is the video frames, the quality of which can be evaluated with Peak-Signal-to-Noise-Ratio (PSNR)~\cite{psnr}. Given a PSNR threshold evaluating the worst quality of image that can tolerate, denoted as $\xi$, then an image indexed $f$ meeting the condition ${\rm PSNR}_f \geq \xi$ can be considered as ``Good", and vice versa.} (denoted as $\cal Q$), the inference error rate is defined as $\epsilon \!=\! g(\mathcal {Q,D})$, where $\epsilon \in [0,1]$, $g(\cdot)$ is a mapping function. In the proposed HMTD framework, the Intersection-over-Union (IoU, denoted as $\omega$) is used to characterize the inference error rate, i.e.,  $\epsilon \!=\! 1-\omega$, where $\omega= \frac{r_D \bigcap r_G}{r_D \bigcup r_G}$, $r_D$ and $r_G$ indicates the detection region and the ground truth region, respectively~\cite{IoU}. Specifically, there is no inference error when $\omega=1$ (i.e., detection region and the ground truth region totally match) while the inference is totally wrong when $\omega=0$. 
\end{definition}

As illustrated in Fig.~\ref{NN_framework}, the number of lower-level layers and higher-level layers is denoted as $l_L$ and $l_H$, respectively. Denote the total layers of the DL model as $K$, we have $l_L\!+\!l_H\!=\!K$. Since it is still a challenge for UAVs to recognize targets from low-quality video frames due to the limited ability of image processing, in this paper, we assume that there exists a certain probability that the fast inference and enhanced inference fail. Without loss of generality, we assume that the image quality may vary from one video frame to another\footnote{The image quality may be affected by the clarity of each video frame and the distance between UAVs and targets~\cite{sensing}.}. As a result, the video image quality maybe sometimes not sufficiently good for the UAVs and the MES to achieve the correct inference. In this case, to improve the inference accuracy and meet the latency demand of the target tracking, the intermediate data from the lower-level layers could be further offloaded to the MES to improve the inference accuracy.
Although it is difficult to obtain an exact analytic formula for the mapping function $\epsilon \!=\! g(\mathcal {Q,D})$ in reality, the observations and conclusions in the paper do {NOT} depend on the exact formula and would not change even the exact formula changes because the trends will remain similar: 
\textbf{good data quality, stronger DL model and more DL layers will have less inference error} \cite{{quality_on_DL01},{quality_on_DL02}}.

\begin{remark}
In the proposed HMTD framework, there exists a \textit{\textbf{trade-off}} between the achieved inference accuracy and the introduced delay. In other words, the two inference modes have merits and shortcomings, i.e., the achievement of a low inference error rate is at the expense of inference delay. For example, when a large portion of the data is with ``Bad" quality, we may not be able to keep the overall inference delay small enough because the inference error rate constraint should also be satisfied. Therefore, the target losing may still occur. This suggests that the UAVs should combine the fast inference with the enhanced inference smartly and allow ample time to ``learn'' the sensing data and take proper actions during challenging environment such as bad weather or the high mobility of the target.
\end{remark}

\subsection{DL Tasks Model}
In the proposed HMTD framework, the UAVs capture video sequence with the embedded camera and then the captured video frames are required to be processed to infer and update the coordinates of the detection region in near-real-time. 

\begin{definition}
\label{D2}
\textit{\textbf{Deep Learning Tasks:}} In this paper, the DL task is defined as the task processed by the pre-trained DL model. The input of DL tasks is the captured images and the output include two folds: target detection given by the classification branch and detection region coordinates inferred by the regression branch, as illustrated in Fig.~\ref{NN_framework}.
\end{definition}

For $U_i$, the DL task can be characterized by a three-tuple of parameters, i.e., ${\cal J}_i(s_i, c_i, \sigma_i)$. Specifically, $s_i$ [bits] denotes the size of computation input data, $c_i$ [cycles] denotes the total number of CPU cycles required to accomplish the computation of $s_i$, and $\sigma_i$ [secs] denotes the maximum tolerable delay.  Due to the environmental changes between UAV and the MES, the wireless channel condition may vary accordingly, which may lead to the unavailability of the wireless channel in some cases. Specifically, if the wireless link is available, then the UAVs can offload the DL tasks to the MES and can also receive the results from the MES via the wireless link. Otherwise,  the DL tasks cannot be offloaded to the MES due to the wireless channel between UAVs and the MES is unavailable (e.g., wireless channel suffers deep fading). 

\begin{table}[t] 
  \captionsetup{font={footnotesize}} 
\caption{Mapping Relationship between Inference Error Rate and Data Quality}
\label{mapping}
\small
\begin{center}
\begin{tabular}{c|c|c}
  \hline
 \diagbox{\textbf{Data Quality}}{\textbf{Error Rate}}  {\textbf{Inference}}
 & Fast & Enhanced\\
 \hline
 Good & $\epsilon_L$ & $0$ \\
 \hline
 Bad  & $1$ & $\epsilon_H$ \\
  \hline
\end{tabular}
\label{tab1}
\end{center}
\end{table}
 
For simplicity, in this paper, suppose that the quality of the video frames captured by the UAV falls into two categories: ``Good" or ``Bad"\footnote{In the UAV tracking context, the video frames' quality could be affected by the distance and the surrounding environment, e.g., the severe weather such as thunderstorm and sand dust, etc.}. The mapping relationship between the inference error rate of the DL model and the input data quality is illustrated in Table~\ref{mapping}\footnote{{The mapping relationship $g(\cdot)$ given in Table~\ref{mapping} may be roughly in practice. To solve this problem, we can trace the mapping relationship via curve-fitting based on the testing experimental results. Specifically, we collect the images data and then evaluate the data quality $\cal Q$ with the PSNR metric. The normalized metric $\cal D$, ${\cal D} \in [0,1]$, is used to evaluate the capability of the DL model, and we suppose that a larger value $\cal D$ indicates a stronger DL model. Then we perform the testing experiments and calculate the inference error rate with varying combinations of $\cal D$ and $\cal Q$. According to the testing results, we could use curve-fitting technique to fit the inference error rate mapping curve and determine the related coefficients correspondingly~\cite{YB_TWC}. With the fitted mapping curve, the inference error rate $\epsilon$ can be inferred continuously once $\cal D$ and $\cal Q$ are given.}}. Specifically, the lower-level layers embedded at the UAV infer the ``Good" frames with an average error rate of $\epsilon_L \in (0, 1)$, while error probability $1$ is assumed when faced with the frames with ``Bad" quality. For the higher-level layers implemented at the MES, the average inference error rate  $\epsilon_H \in (0,1)$ is achieved when processing the ``Bad" frames while no error occurring is assumed on processing the ``Good" frames. 

\subsection{Communication Model}
A three-dimensional Cartesian coordinate system is used to characterize the communication link between the UAV in aerial and the MES on the ground, as illustrated in Fig.~\ref{fig3}. It is assumed that the MES is located at position $\textbf{\textit{p}}_M=[x_M, y_M]^T\in \mathbb{R}^{2\times 1}$, the UAV flies along a horizon trajectory with a fixed altitude $H$ and communicates with its associated BS in a time-division manner. In this case, the position of $U_i$ can be characterized by the discrete-time locations, i.e., $\textbf{\textit{p}}^i_U=[ x^i_U, y^i_U ]^T\in \mathbb{R}^{2\times 1}$. As the altitude of the UAV is much higher than that of the MES on the ground, it is reasonably assumed that the communication channels between MES and UAV are dominated by line-of-sight (LOS)\cite{UAV_MEC_1}. In this case, it is reasonably assumed that the channel condition does not change within each offloading procedure. The channel gain between $U_i$ and the MES can be obtained as
\begin{equation}\label{e01}
{\cal H}_i=h_0 \lambda_i^{-2}=\frac{h_0}{H^2+\left \| \textbf{\textit{p}}_M - \textbf{\textit{p}}^i_U \right \|^2}, \ i\in {\bf N},
\end{equation}
where $h_0$ represents the channel gain at the reference distance of $1\ m$, $\lambda_i$ denotes the distance between $U_i$ and the MES.



\begin{figure}[t]
  \captionsetup{font={footnotesize }}
\centerline{ \includegraphics[width=2.25in,height=1.65in]{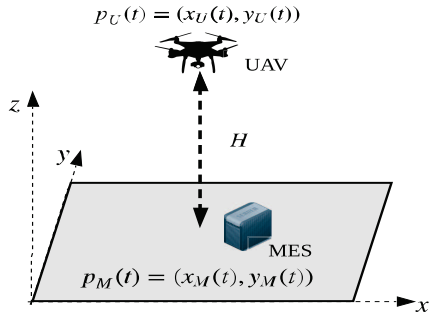}}
\caption{Communication model between UAVs and the MES.}
\label{fig3}
\end{figure}

We consider the MEC system with OMA (e.g., OFDMA) as the multiple access scheme in the offloading, in which the UAVs offload their DL tasks to the MES via orthogonal sub-bands simultaneously\footnote{ In the cellular IoVs based visual target tracking  system, the uplink transmission of the intermediate data from the lower-level layers to the higher-level layers dominates~\cite{zhang}. Therefore, in this article, we focus on the uplink transmission from the UAVs to the MES. Although an OMA scenario is assumed in this paper, our proposed offloading framework can be also extended into scenarios using more advanced non-orthogonal multiple access (NOMA) with a minor modification on the communication model.}. Denote the noise power as the white Gaussian noise with zero mean and variance $\chi^2$, $P_t$ is the transmission power of the UAVs, then the received signal-to-interference-plus-noise ratio (SINR) at the MES can be calculated as
$\frac{P_t |{\cal H}_i|^2}{\chi^2}$.
Therefore, the available transmission rate between $U_i$ and the MES can be calculated as
\begin{equation}\label{e02}
R_i=\frac{B}{n} log_2{\left ( 1+ \frac{P_t |{\cal H}_i|^2}{\chi^2} \right )}, \ i \in {\bf N},
\end{equation}
where $B$ stands for the transmission bandwidth between the UAV and MES, which can be further divided into $n$ sub-bands for the offloading communication.


Without loss of generality, two kinds of offloading strategies, i.e., binary offloading and partial offloading, are considered to optimize the proposed offloading framework, which is detailed in Section~\ref{binary_offloading_framework} and Section~\ref{partial_offloading_framwork}, respectively.

\section{Optimization for Binary Offloading Framework}\label{binary_offloading_framework}
\subsection{Problem Formulation}
For the binary offloading scheme, the offloading probability is first introduced as below.

\begin{definition}
\label{offloading_p}
 \textit{\textbf{Offloading Probability:}}  For the $i$-th offloading UAV $U_i$, the offloading probability ($\mu_i$) is defined as the probability that the UAV offloading the tasks to the MES. Suppose that the offloading UAVs can evaluate the inference error rate obtained by the DL lower-layers in near-real-time~\cite{Vehicle01}, denoted as ${\epsilon_L}$. If ${\epsilon_L}$ is above a certain threshold $\epsilon_T^i$, then the data needs to be offloaded to the MES to guarantee the inference accuracy. Thus, the offloading probability of $U_i$ equals the probability that $\epsilon_L \geq \epsilon_T^i$, i.e., 
\begin{equation}\label{P_offloading} 
\mu_i= {\rm{Pr}}\left \{\epsilon_L \geq \epsilon_T^i \right \} =\int_{\epsilon_{\rm th}^i}^{\infty } e^{-x}dx=e^{-\epsilon_T^i}.
\end{equation}
\end{definition}

\subsubsection{\textbf{Local Computing}}
Denote $f^i_{l}$ as the CPU-cycle frequency (i.e., CPU cycles per second) of $U_i$, $i \in {\bf N}$, the local computation delay is calculated as 
\begin{equation}\label{binary_local_delay}
\tau^i_{l} = \frac{c_i}{f_l^i}+ \rho_i  \left (\left (1-\eta\right )\epsilon_L + \eta \right ),
\end{equation}
where $\eta$ denotes the probability that the DL tasks data is considered as ``Bad", $1-\eta$ denotes the probability that the data is considered as ``Good". $\rho_i$ is introduced as all of the DL tasks failing and dropping penalty of delay\footnote{Suppose that the UAV may lose the target due to the tasks processing latency or the fast moving speed of the target. This is because the target might appear easily outside of the searching region used for the tracking, and enlarging the searching region can cause the delay and energy performance degradation to the visual target tracking.}, which is generally no smaller than the tasks processing delay, i.e., $\rho_i > {\rm max} \{\frac{c_i}{f_l^i}, \frac{c_i}{f_i} \}$, $f_i$ denotes the allocated CPU computation resource to $U_i$ by the MES.

According to the widely adopted model of the energy consumption\cite{k1}, the energy consumption processing ${\cal J}_i$ locally with the CPU clock speed $f _{l}^{i}$ can be calculated as
\begin{equation}\label{binary_local_energy}
\varepsilon_{l}^{i}=\kappa \left ( f_{l}^{i} \right )^2  c_i + \xi_i  \left (\left (1-\eta\right )\epsilon_L + \eta \right ),
\end{equation}
where $\kappa$ denotes the energy efficiency parameter that is mainly depends on the chip architecture\cite{k}, $\xi_i$ is introduced as all of the DL tasks failing and dropping penalty of energy consumption.

Based on (\ref{binary_local_delay}) and (\ref{binary_local_energy}), the weighted-cost for computing ${\cal J}_i$ locally for binary offloading scheme is achieved as
\begin{equation}\label{binary_local_cost}
{\cal O} _{l}^{i}=\theta \tau _{l}^{i} + \left(1-\theta \right) \varepsilon_{l}^{i},
\end{equation}
where $\theta$, $0\le \theta\le 1$, specifies the UAV's preference on processing delay, while $1-\theta$ specifies the UAV's preference on energy consumption. 

Substituting (\ref{binary_local_delay}) and (\ref{binary_local_energy}) into (\ref{binary_local_cost}), ${\cal O} _{l}^{i}$ is derived as 
\begin{equation}
\label{binary_local_cost_1}
{\cal O} _{l}^{i}= \theta\frac{c_i}{f^i_l}+\left(1-\theta \right) \kappa \left ( f_{l}^{i} \right )^2 c_i + 
 \tilde{\epsilon_L} \left (\theta \rho_i + \left(1-\theta \right) \xi_i \right),
\end{equation}
where $\tilde{\epsilon_L}=\left (1-\eta\right )\epsilon_L + \eta$.

\begin{remark}
An UAV with short battery life is prone to decrease the coefficient $\theta$ so as to save more energy at the expense of longer tasks processing delay, and vice versa.
\end{remark}


\subsubsection{\textbf{Offloading Computing}} For the offloading computing, in case that $U_i$ offloads ${\cal J}_i$ to the MES, the incurred delay and energy consumption comprise the following two items\footnote{Since the size of the execution results is generally much smaller compared to that of input data, so the corresponding delay and energy consumption is ignored\cite{Full-1}.}: (1) the delay and energy offloading ${\cal J}_i$ to the MES via the wireless link, and (2) the delay and energy executing ${\cal J}_i$ at the MES. 
Suppose that the MES can provide computation offloading service to multiple UAVs concurrently, the queuing delay at the MES is ignored~\cite{CPU_parallel}. 
During the execution of the tasks at the MES, the computation resources available at the MES are shared among the associating UAVs and quantified by the allocated computational resources expressed in terms of the number of CPU cycles-per-second
, i.e., $f_{i}, \forall i \in {\bf N}$. The computing resource constraint should be satisfied, which is expressed as
$\sum_{i\in{\bf N}}^{ }f_{i} \leq F$, where $F$ denotes the entire computational resources of the MES.

Therefore, the delay for offloading the task  ${\cal J}_i$ to the MES is given by
 \begin{equation}\label{binary offloading_delay}
 \tau^i_{o} = \frac{c_i}{f_l^i} +\gamma_i \left (  \frac{s_i}{R_i}+\frac{c_i}{f_i} \right )+ \rho_i \eta {\epsilon}_H,
\end{equation}
where $\gamma_i$ denotes the scale coefficient of data size output from the lower-level layers of $U_i$, i.e., $\gamma_i=\frac{s^i_{out}}{s_i}$, $s^i_{out}$ is the data size output from $U_i$.

The energy consumption of $U_i$ using offloading computing is calculated as 
\begin{equation}\label{binary_offloading_energy}
\varepsilon_{o}^{i}=\kappa \left ( f_{l}^{i} \right )^2 c_i +\gamma_i \left (P_t^i  \frac{s_i}{R_i}+ P^i_I \frac{c_i}{f_i} \right )+ \xi_i \eta {\epsilon}_H,
\end{equation}
where $P_t^i$ is the transmission power of $U_i$. $P_t^i$ denotes the power consumption of $U_i$ staying idle while waiting for the execution results from the MES.

According to (\ref{binary offloading_delay}) and (\ref{binary_offloading_energy}), the weighted-cost for offloading ${\cal J}_i$ to the MES is given by
\begin{equation}\label{binary_cost_offloading}
{\cal O}_{o}^{i}=\theta  \tau_{o}^{i} + (1-\theta) \varepsilon_{o}^{i}, \  \forall i \in {\bf N}.
\end{equation}

Substituting (\ref{binary offloading_delay}) and (\ref{binary_offloading_energy}) into (\ref{binary_cost_offloading}), ${\cal O}_{o}^{i}$ can be rewritten as 
\begin{equation}
\label{binary_cost_offloading1}
\begin{aligned}
{\cal O} _{o}^{i}&\!=\! \theta \frac{c_i}{f_l^i} \!+\! \left(1-\theta \right) \kappa \left ( f_{l}^{i} \right )^2 c_i \!+\! \tilde{\epsilon_H} \left(\theta \rho_i \!+\! (1-\theta)\xi_i \right) \!+\! \gamma_i \left (\frac{s_i}{R_i} \left(\theta\!+\!(1\!-\!\theta)P^i_t \right)\!+\!\frac{c_i}{f_i} \left(\theta \!+\! (1\!-\!\theta)P^i_I \right) \right ),
\end{aligned}
\end{equation}
where $\tilde{\epsilon_H}=\eta \epsilon_H$.

Based on (\ref{P_offloading}), (\ref{binary_local_cost_1}) and (\ref{binary_cost_offloading1}), the overall cost of $U_i$ using binary offloading scheme is obtained as
\setcounter{equation}{11}  
\begin{equation}\label{binary_overall_cost}
{\cal O}_i^B=\left(1-\mu_i\right){\cal O}_l^i+\mu_i {\cal O}_o^i={\cal O}_l^i+\mu_i \left({\cal O}_o^i-{\cal O}_l^i\right).
\end{equation}

The weighted-sum cost of all the offloading UAVs is calculated as
\begin{equation}\label{binary_total_cost}
{\cal O}_{total}^B=\sum_{i \in{\bf N}} \left(1-\mu_i\right){\cal O}_l^i+\mu_i {\cal O}_o^i,
\end{equation}
with ${\cal O}_{l}^{i}$ and ${\cal O}_{o}^{i}$ defined in (\ref{binary_local_cost_1}) and (\ref{binary_cost_offloading1}), respectively, and $0 \le \mu_i \le 1$ specifying the offloading probability of $U_i$.

Given the binary offloading system model described previously, our goal is to develop an optimal offloading probability (denoted as $\mu_i^*$) for UAVs to minimize the total weighted-sum cost. Here, we formulate the optimal offloading as a weighted-sum cost minimization problem (denoted as $\mathscr{P}_B$), subject to individual UAV's delay and power supply constraints and the computational resource limit of the MES.

$\mathscr{P}_B$ ({\textit{Binary Offloading Problem}}):
\begin{subequations} \label{binary_original_problem}
\begin{align}
& \;\;\;\;\underset{\mu_i \in [0,1]}{\rm minimize}\;\; {\cal O}_{total}^B \notag \\
& \;\;\;\;\;\;{\rm{s}}{\rm{.t}}{\rm{.}}\;\;\;\mathbf{C1}: \; \epsilon_i \le \epsilon_{T}^i, \ \forall i \in {\bf N}, \\
& \;\;\;\;\;\;\;\;\;\;\;\;\;\;\mathbf{C2}: \; 0 \le f_{i} \le F, \ \forall i \in {\bf N}, \\
& \;\;\;\;\;\;\;\;\;\;\;\;\;\;\mathbf{C3}: \; \sum_{i=1}^{n} f_{i} \le F.
 \end{align}
\end{subequations}

The constraints in the formulation of $\mathscr{P}_B$ are detailed as follows. $\mathbf{C1}$ makes sure that the average inference error rate processing ${\cal J}_i$ should not exceed the maximum tolerable threshold. $\mathbf{C2}$ and $\mathbf{C3}$ guarantee that the computational resource allocated to $U_i$ and the sum of the computational resources allocated to all the offloading UAVs should not exceed the computation resource limit of the MES. 

In the following, we analyze the optimal offloading solutions of the binary offloading problem based on the availability of wireless channel between UAVs and the MES, as described in the following parts B-C.

\subsection{Wireless Channel is Unavailable}
In the UAV tracking system, when the wireless channel between the UAV and the MES is under the thunderstorms and other extreme weather conditions, the wireless channel could be unavailable. In this case, all the DL tasks can be only processed at the UAVs. 
The overall cost of $U_i$ and the total cost using local computing can be rewritten as ${\cal O}_i^B={\cal O}_l^i$ and  ${\cal O}_{total}^B=\sum_{i \in{\bf N}} {\cal O}_l^i$. Therefore, the problem $\mathscr{P}_B$ can be reformulated as $\mathscr{P}_B^1$.

$\mathscr{P}_B^l$ \textit{(Binary - Local Computing Problem)}:
\begin{subequations} \label{e18}
\begin{align}
& \;\;\;\;\underset{\mu_i=0}{\rm minimize}\;\; \sum_{i \in{\bf N}} {\cal O}_l^i \notag \\
& \;\;\;\;\;\;{\rm{s}}{\rm{.t}}{\rm{.}}\;\;\;\mathbf{C1.1}: \; \theta, \eta \in [0,1], \\
& \;\;\;\;\;\;\;\;\;\;\;\;\;\;\mathbf{C1.2}: \; \epsilon_L \in (0,1),
 \end{align}
\end{subequations}
where the condition 
$\rm{C1.1}$ accounts for the range of the weight coefficient and the ``Bad" data probability. The condition $\rm{C1.2}$ specifies the range of inference error rate at UAVs with ``Good" input data. 

\begin{figure}[t]
\centering
 \captionsetup{font={footnotesize }}
\subfigure[]{
\includegraphics[width=2.25in,height=1.55in]{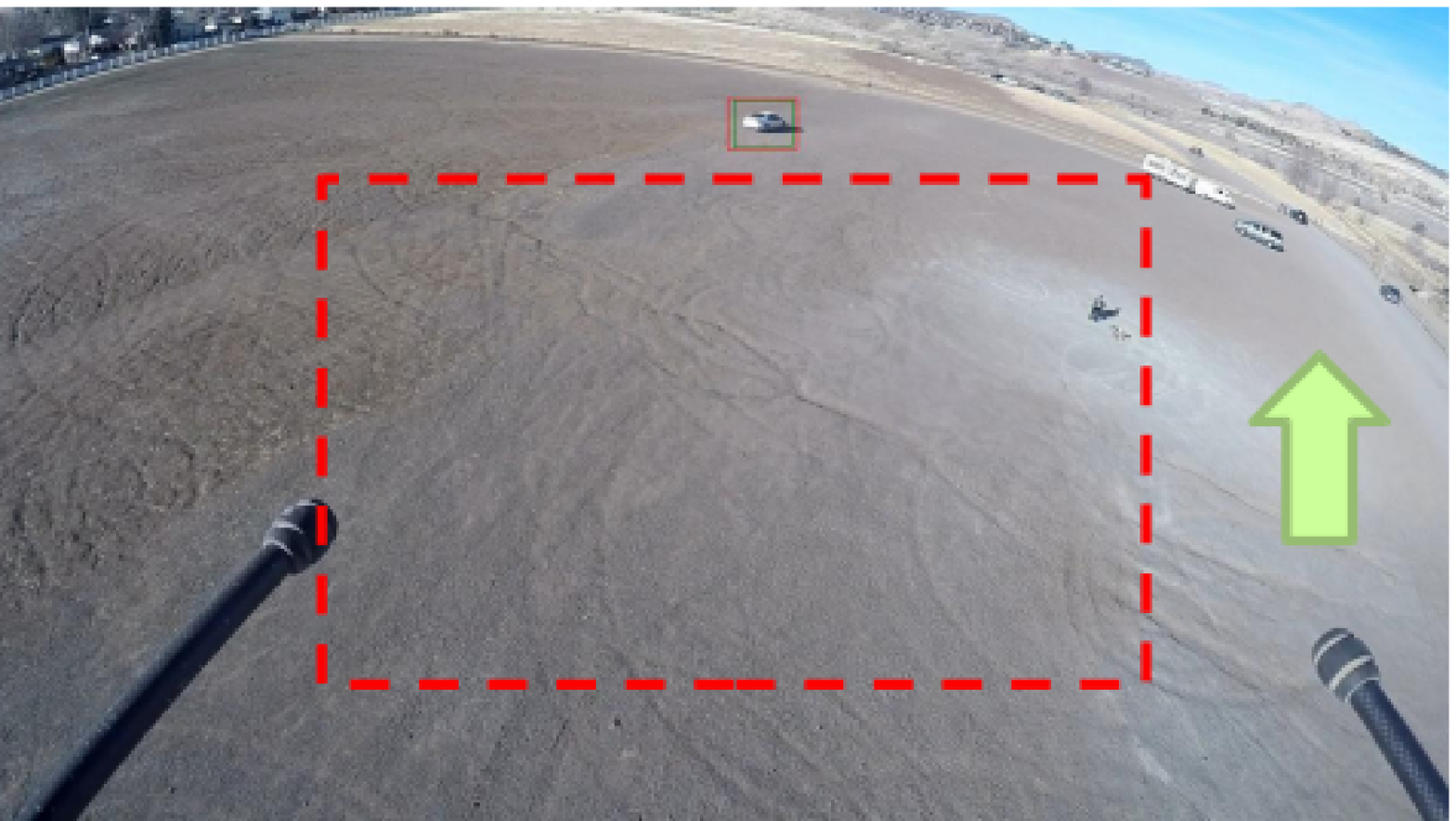}}
\hspace{-0.1in}
\subfigure[]{
\includegraphics[width=2.25in,height=1.55in]{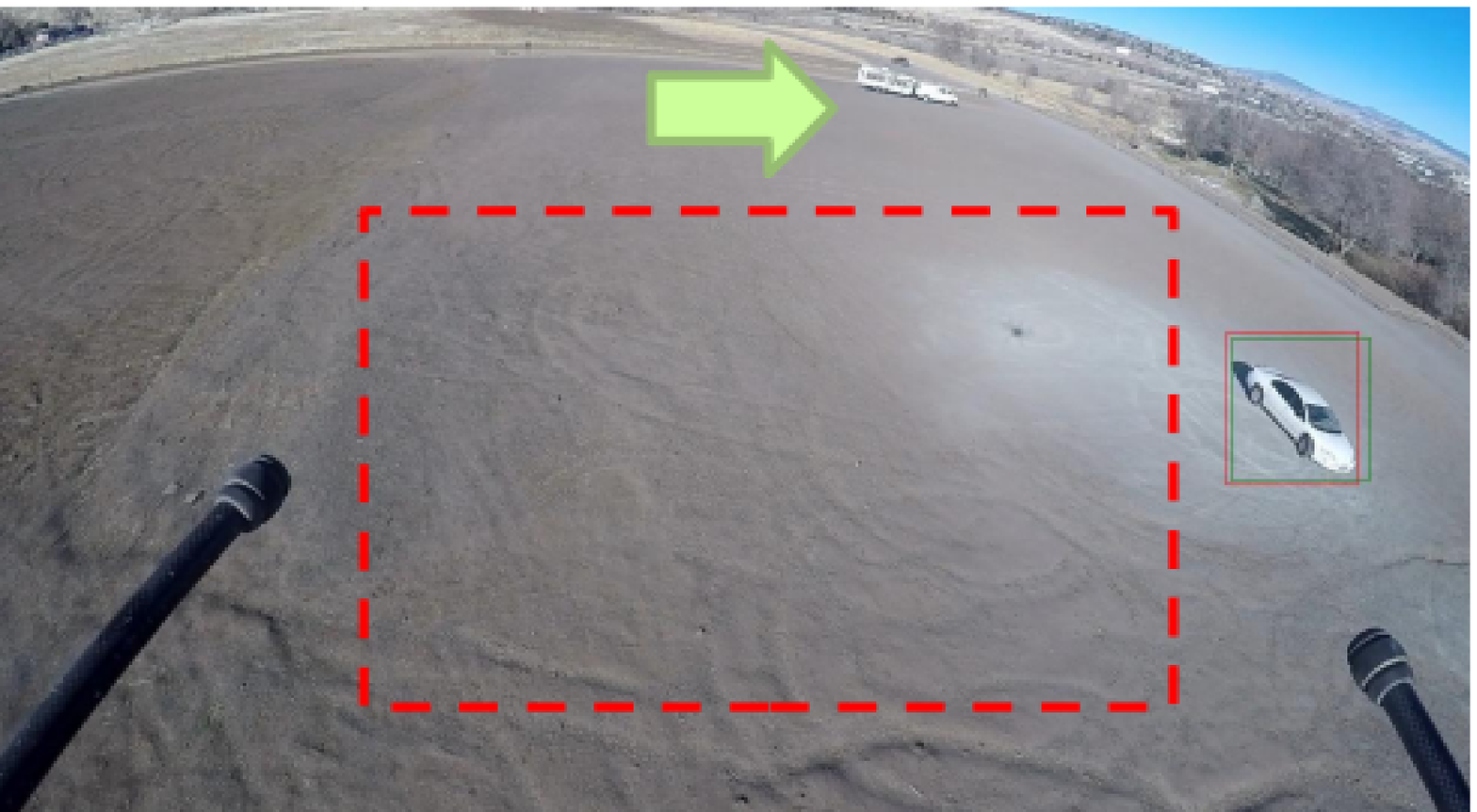}}
\caption{The target (e.g., a white car) is detected outside of the detection region, the UAV would move forward (shown in (a)) or move right (shown in (b)), and try to make the target be inside the detection region.} 
\label{detect}
\end{figure}
 
Owing that the UAV's inference error rate is obtained as $\epsilon_L=g(\mathcal {D}, \eta)$. Once a DL model $\cal D$ and ``Bad" data percentage $\eta$ are given, $\epsilon_L$ can be calculated accordingly. It is observed from (\ref{binary_local_cost_1}) that once $\eta$ and $\theta$ are given, the local tasks processing cost (${\cal O}_l^i$) is proportional to the tasks dropping penalty ($\rho_i$ and $\xi_i$), which is mainly determined by the object detecting and navigation overhead of the UAV.

\begin{remark}
During the target tracking, there exists a {trade-off} between the detection region size and the weight-sum cost of UAVs. Specifically, a detection region is set up that if the object appears inside, the UAV would stay around without taking any action. If the size of the region is too small, it is easy that the object disappears from view. As a result, the UAV needs to relocate the target, which introduces additional penalty costs. On the other hand, if the size is too large, the UAV would cause more waste of computing power and battery to search the target. 
\end{remark}

During the target tracking, if the target is detected outside of the {detection region}, the UAV would make adjustments and try to make the target detected inside the region, as illustrated in Fig.~\ref{detect}. On the contrary, if the target is detected still inside of the detection region, the UAV would stay around without taking too many actions. Suppose that the size of the detection region is fixed, the tasks dropping penalty can be decreased when the target moves slowly.

\subsection{Wireless Channel is Available}
When the wireless channel between the UAV and the MES is available, the DL tasks can be either processed totally at the UAVs or further offloaded to the MES. The overall cost of $U_i$ using binary offloading scheme is given by ${\cal O}_i^B={\cal O}_l^i+\mu_i \triangle {\cal O}_i$, where $\triangle{\cal O}_i = {\cal O}_o^i-{\cal O}_l^i$. In this case, the problem $\mathscr{P}_B$ can be reformulated as $\mathscr{P}_B^2$.

$\mathscr{P}_B^2$ ({\textit{Binary - Offloading Computing Problem}}):
\begin{equation} \label{binary_offloading_problem}
\begin{aligned}
& \;\;\;\;\underset{\mu_i \in [0,1]}{\rm minimize}\;\; \sum_{i \in{\bf N}} {\cal O}_l^i+\mu_i \triangle {\cal O}_i  \\
& \;\;\;\;\;\;{\rm{s}}{\rm{.t}}{\rm{.}}\;\;\;\mathbf{C1-C3}.
 \end{aligned}
\end{equation}


In order to derive the optimal offloading probability (denoted as $\mu_i^*$) to minimize the overall cost of $\mathscr{P}_B^2$, $\triangle {\cal O}_i$ is analyzed first.

Substituting (\ref{binary_local_cost_1}) and (\ref{binary_cost_offloading1}) into $\triangle {\cal O}_i$, we have
\begin{equation}\label{delta_O}
\triangle {\cal O}_i = \gamma_i C_i+ \left(\tilde{\epsilon_H} -\tilde{\epsilon_L} \right) \left(\theta \rho_i+(1-\theta)\xi_i \right),
\end{equation}
where $C_i \!=\! \frac{s_i}{R_i} \left(\theta+(1-\theta)P^i_t \right)\!+\!\frac{c_i}{f_i} \left(\theta + (1-\theta)P^i_I\right)$ indicates the weight-sum cost including delay and energy consumption.

For clarity, let $\alpha\!=\! \tilde{\epsilon_H} - \tilde{\epsilon_L}\!=\!\eta {\epsilon}_H - \left (\left (1-\eta\right )\epsilon_L + \eta \right)$, we can obtain $\frac{\alpha}{\eta}\!=\! {\epsilon}_H - \left ( \frac{1-\eta}{\eta} \epsilon_L + 1 \right)$, which is negative. This indicates that the average inference error rate at UAV is larger than that of MES. Let $\gamma_i^{B}=\frac{\left(\tilde{\epsilon_L} - \tilde{\epsilon_H} \right) \left(\theta \rho_i+(1-\theta)\xi_i \right)}{C_i}$, then we derive the optimal offloading probability to minimize ${\cal O}_i^B$ in the following two cases, as illustrated in Fig.~\ref{solution}.


\begin{itemize}
\item Case 1: $\gamma_i < \gamma_i^{B}$. In this case, $\triangle {\cal O}_i <0$ holds and thus ${\cal O}_i^B$ varies inversely with the offloading probability $\mu_i$.

\item Case 2: $\gamma_i > \gamma_i^{B}$.  In this case, $\triangle {\cal O}_i >0$ is achieved. Therefore, ${\cal O}_i^B$ varies proportionally to the offloading probability $\mu_i$.
\end{itemize}

\begin{figure}[t]
 \captionsetup{font={footnotesize }}
\centerline{ \includegraphics[width=2.55in, height=1.35in]{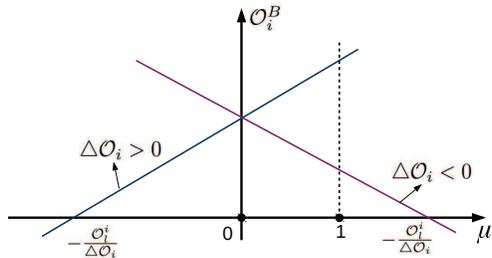}}
\caption{Offloading probability ($\mu_i$) versus the overall system cost (${\cal O}_i^B$) using binary offloading scheme.}
\label{solution}
\end{figure}

According to Table~\ref{mapping}, the average inference error rate of $U_i$ using binary offloading scheme can be achieved as 
\begin{equation}\label{binary_inference_error}
\epsilon_i = \left(1-\mu_i\right)\tilde{\epsilon_L} + \mu_i \tilde{\epsilon_H} = \tilde{\epsilon_L} - \mu_i \left(\tilde{\epsilon_L} - \tilde{\epsilon_H}  \right).
\end{equation}
 
Since $\tilde{\epsilon_L} > \tilde{\epsilon_H}$ generally holds, substituting (\ref{binary_inference_error}) into (\ref{binary_original_problem}a), we can achieve $\mu_i \ge \tilde{\mu_i}$, where $\tilde{\mu_i} \!=\! \cfrac{\tilde{\epsilon_L}-\epsilon_T^i}{\tilde{\epsilon_L}-\tilde{\epsilon_H}} \!\in\! [0, 1]$ is achieved because $\epsilon_T^i \in [\tilde{\epsilon_H},\tilde{\epsilon_L}]$ generally holds. 
According to Fig.~\ref{solution}, in order to minimize the cost ${\cal O}^B_i$, we can obtain the optimal offloading probability in the two cases above as follows: In case 1, since ${\cal O}_i^B$ varies inversely with $\mu_i$, we can obtain $\mu_i^*={\rm max}\{\tilde{\mu_i},1 \}=1$. In case 2, since ${\cal O}_i^B$ varies proportionally to $\mu_i$, we can obtain $\mu_i^*={\rm max}\{0, \tilde{\mu_i} \}=\tilde{\mu_i}$.

Therefore, the optimal offloading probability for the binary offloading scheme is calculated as 
\begin{equation} \label{binary_optimal_p}
\mu_i^*=\left\{\begin{matrix}
1, & {\rm If\ } \gamma_i < \gamma_i^{B}, \ (\text{Con. A}) \\ 
 \tilde{\mu_i}, & \text {If} \ \gamma_i > \gamma_i^{B}, \ (\text{Con. B})
\end{matrix}\right..
\end{equation}

\begin{remark}
It is observed from (\ref{binary_optimal_p}) that if $\gamma_i$ is small, e.g., $\gamma_i < \gamma_i^{B}$ denoting the size of the intermediate data output from the lower-level layer (denoted as $s_{out}^i$) is relatively small, so $\mu_i^*$ is prone to be as large as possible (i.e., $\mu_i^*=1$) to improve the inference accuracy without introducing too much communication delay. On the contrary, once $\gamma_i > \gamma_i^{B}$ holds, which indicates that $s_{out}^i$ could be very large compared to the size of the original data. In this situation, offloading the DL tasks to the MES will introduce too much communication delay and thus $\mu_i^*$ should not be too large, i.e., $\mu_i^*=\tilde{\mu_i}$. Furthermore, it is observed that $\mu_i^*$ is in inverse proportion to the inference error threshold $\epsilon_T^i$. That is to say, the larger $\epsilon_T^i$ the smaller $\mu_i^*$, which indicates that the DL tasks are prone to be processed locally. On the contrary, the limited computation resource of the UAV becomes a bottleneck. Therefore, it is straightforward to offload DL tasks to fully explore the computation resources of the MES.
\end{remark}

\section{Optimization for Partial Offloading Framework}\label{partial_offloading_framwork}
 In this section, the partial offloading model is investigated, which is more general in practice in that it can fully utilize the computation resources in both of the UAV and the MES. 
In the following, we first formulate the latency-minimization problem as a piecewise-convex problem and then derive the optimal offloading ratio (i.e., $\beta_i^*$). Then, two cases are considered where we analyze $\beta_i^*$ to minimize energy consumption and total cost, respectively. Finally, a special scenario is considered assuming that the UAV can distribute the ML tasks according to the data quality $\cal Q$, and the optimal tasks segmentation ratio is derived accordingly.

\subsection{Problem Formulation}
For the partial offloading scheme, the offloading ratio is first defined as below.

\begin{definition}
\label{offloading_ratio}
 \textbf{Offloading Ratio:} For $U_i$, the offloading ratio ($\beta_i$) is defined as the ratio (or portion) of the DL tasks that is offloaded to the MES. Therefore, denote $\beta_i$ as the ratio of data offloaded to the MES of $U_i$, whereas, $1-\beta_i$ indicates the ratio of data to be processed locally. Suppose that the time-interdependency between each video frame within the DL task is ignored, then ${\cal J}_i$ can be divided into two parts, i.e., $\beta_i s_i$ (bits) is offloaded to the MES while $(1-\beta_i)s_i$ (bits) is processed locally at the UAV.
 \end{definition}

 
\subsubsection{\textbf{Local Computing}} For the local computing, $(1-\beta_i) s_i$ of ${\cal J}_i$ is processed at the UAV. Let $\tilde{\epsilon_L}=\left (1-\eta\right )\epsilon_L + \eta$, 
the local computation delay and energy consumption is respectively calculated as 
\begin{equation}\label{partial_local_delay}
\tau^i_{l} = \left ( 1-\beta_i \right ) \left( \frac{c_i}{f_l^i}+ \rho_i  \tilde{\epsilon_L} \right),
\end{equation}
and
\begin{equation}\label{partial_local_energy}
\varepsilon_{l}^{i}=\left ( 1-\beta_i \right ) \left( \kappa \left ( f_{l}^{i} \right )^2  c_i + \xi_i  \tilde{\epsilon_L} \right).
\end{equation}

Based on (\ref{partial_local_delay}) and (\ref{partial_local_energy}), the weighted-cost for the local computing is calculated as
\begin{equation}\label{local_cost}
{\cal O} _{l}^{i}=\theta \tau _{l}^{i} + \left(1-\theta \right) \varepsilon_{l}^{i}, \ \forall i \in {\bf N}.
\end{equation}

Substituting (\ref{partial_local_delay}) and (\ref{partial_local_energy}) into (\ref{local_cost}), the weighted-cost for computing ${\cal J}_i$ locally can be rewritten as 
\begin{equation}\label{partial_local_cost}
\begin{aligned}
{\cal O} _{l}^{i}&=\left ( 1-\beta_i \right )  \left(\frac{\theta c_i}{f^i_l}+\left(1-\theta \right) \kappa \left ( f_{l}^{i} \right )^2 c_i \right) +  \tilde{\epsilon_L} \left (\theta \rho_i + \left(1-\theta \right) \xi_i \right).
\end{aligned}
\end{equation}

\subsubsection{\textbf{Offloading Computing}}   
For the offloading computing,  $\beta_i s_i$ of ${\cal J}_i$ is offloaded to the MES. Let  $\tilde{\epsilon_H}=\eta \epsilon_H$, the total delay and energy consumption introduced by the offloading computing is given by
\setcounter{equation}{23}
 \begin{equation}\label{partial_delay_offloading}
 \tau^i_{o} = \beta_i\left ( \frac{c_i}{f_l^i} +\gamma_i \left (  \frac{s_i}{R_i}+\frac{c_i}{f_i} \right ) + \rho_i \tilde{\epsilon_H}\right ),
\end{equation}

\begin{equation}\label{partial_energy_offloading}
\varepsilon_{o}^{i}\!=\! \beta_i\left (\kappa \left ( f_{l}^{i} \right )^2 c_i \!+\!\gamma_i \left (P_t^i  \frac{s_i}{R_i}\!+\! P^i_I \frac{c_i}{f_i} \right )  + \xi_i \tilde{\epsilon_H}\right ),
\end{equation}
where $\gamma_i$ is the scale coefficient of data size output from the lower-level layers of $U_i$, and $P_t^i $ is the UAV's transmission power. Suppose that UAVs staying idle while waiting for the execution results from the MES and the power consumption of $U_i$ staying the idle state is $P_{I}^{i}$.

According to (\ref{partial_delay_offloading}) and (\ref{partial_energy_offloading}), the weighted-cost for the offloading computing is given by
\begin{equation}\label{cost_offloading}
{\cal O}_{o}^{i}=\theta  \tau_{o}^{i} + (1-\theta) \varepsilon_{o}^{i}, \  \forall i \in {\bf N}.
\end{equation}

Substituting (\ref{partial_delay_offloading}) and (\ref{partial_energy_offloading}) into (\ref{cost_offloading}), the weighted-cost for computing ${\cal J}_i$ locally can be calculated accordingly. 

\begin{remark}
\label{R5}
Different from the offloading framework where the deep learning models with different capabilities are deployed at the sensing devices and the edge server in IIoT~\cite{IoT_YB}, in the proposed HMTD framework, the lower-level layers and the higher-level layers of the same trained CNN model are deployed at the UAV and MES, respectively. The data privacy preserving is achieved since only the intermediate features are offloaded from the UAVs to the MES. In this context, the total inference delay at the MES not only includes the communication and computing delay but also includes the tasks processing delay at the UAV.
\end{remark}

Based on Remark~\ref{R5}, the total delay introduced by partial offloading scheme is calculated as
\setcounter{equation}{27}
\begin{equation}\label{partial_overall_delay}
\tau_i^P=\frac{c_i}{f_l^i} + {\rm max} \{\delta_p^i,  \delta_{op}^i  \},
\end{equation}
where $\delta_p^i = \left(1-\beta_i \right) \rho_i \tilde{\epsilon_L}$ denoting the delay penalty processing $\left(1-\beta_i \right)$ of ${\cal J}_i$ at the UAV. $\delta_{op}^i=\beta_i \left( \gamma_i \left (  \frac{s_i}{R_i}+\frac{c_i}{f_i} \right)  + \rho_i \tilde{\epsilon_H} \right )$ including the transmission delay of intermediate data, processing delay at the MES, and the delay penalty processing $\beta_i$ of ${\cal J}_i$ at the MES.

Based on the analysis above, the overall cost of $U_i$ using partial offloading scheme is obtained as
\begin{equation}\label{partial_overall_cost}
{\cal O}_i^P=\theta \left( \frac{c_i}{f_l^i} + {\rm max} \{\delta_p^i,  \delta_{op}^i  \} \right) + \left(1-\theta \right) \left(\varepsilon_{l}^{i} + \varepsilon_{o}^{i} \right).
\end{equation}

Therefore, the weighted-sum cost of all the offloading UAVs using partial offloading scheme is calculated as
\begin{equation}\label{partial_total_cost}
{\cal O}_{total}^P = \sum_{i \in{\bf N}} {\cal O}_i^P,
\end{equation}
where ${\cal O}_i^P$ is given in (\ref{partial_overall_cost}).

%
%

Given the partial offloading system model described previously, our goal is to develop an optimal offloading ratio (denoted as $\beta_i^*$) for UAVs to minimize the total weighted-sum cost combining execution delay and energy consumption of all
UAVs under the constrain of maximum tolerable inference error rate.
In this case, we formulate the cost minimization as a piecewise-convex optimization problem ($\mathscr{P}_P$).

$\mathscr{P}_P$ ({\textit{Partial Offloading problem}}):
\begin{equation} \label{partial_offloading_problem}
\begin{aligned}
& \;\;\;\;\underset{\beta_i \in [0,1]}{\rm minimize}\;\; {\cal O}_{total}^P  \\
& \;\;\;\;\;\;{\rm{s}}{\rm{.t}}{\rm{.}}\;\;\;\mathbf{C1-C3},
 \end{aligned}
\end{equation}
with ${\cal O}_{total}^P$ defined in (\ref{partial_total_cost}).

In the following, the closed-form expressions for $\beta_i^*$ is devised in different scenarios with specific objectives. 

\subsection{Case 1: Delay-Sensitive Objective}
Suppose that the power consumption is not a critical concern for UAVs since some technologies can be introduced to provide convenient and sustainable energy supply to the UAVs, e.g., wireless power transfer~\cite{charge_UAV01} and laser-beamed power supply~\cite{charge_UAV02}. In this case, the delay introduced (including computing delay and communication delay) is the main concern, i.e., $\theta=1$. The total system cost can be simplified as 
\begin{equation}
{\cal O}_{total}^{P_{case1}} =\sum_{i=1}^{n} \left( \frac{c_i}{f_l^i} + {\rm max} \{\delta_p^i,  \delta_{op}^i  \} \right).
\end{equation}
 
Therefore, $\mathscr{P}_P$ can be transformed as $\mathscr{P}_P^1$, which is described as follows.
 
$\mathscr{P}_P^1$ (\textit{Partial - Delay Minimization  Problem}):
\begin{equation}
\begin{aligned} \label{delay_problem}
&\underset{\beta_i}{\rm minimize}\;\; \sum_{i=1}^{n} \left( \frac{c_i}{f_l^i} + {\rm max} \{\delta_p^i,  \delta_{op}^i  \} \right) \\
& \;\;\;\;\;\;{\rm{s}}{\rm{.t}}{\rm{.}}\;\;\; \mathbf{C1-C3}.
 \end{aligned}
 \end{equation}
  
 \begin{figure}[t]
  \captionsetup{font={footnotesize }}
\centerline{ \includegraphics[width=2.55in, height=1.5in]{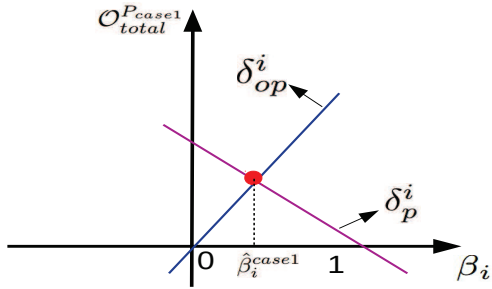}}
\caption{Offloading ratio ($\beta_i$) versus the overall system cost (${\cal O}_{total}^{P_{case1}}$) using partial offloading scheme in case 1.}
\label{optimal_beta}
\end{figure}

To analyze the problem $\mathscr{P}_P^1$, the following Lemma 1 is introduced. 

\newtheorem{lemma}{\textit{Lemma}}
\begin{lemma}
When $\eta$, $\epsilon_L$ and $\epsilon_H$ are given, the offloading ratio for $\mathscr{P}_P^1$ can be obtained as
\begin{equation}\label{optimal_beta_c1}
\centering
\hat{\beta}_i^{case1}=\cfrac{\rho_i \tilde{\epsilon_L} }{\gamma_i \left(\frac{s_i}{R_i}+\frac{c_i}{f_i} \right) + \rho_i \left(\tilde{\epsilon_L} + \tilde{\epsilon_H} \right)}.
\end{equation}
\end{lemma}

\begin{proof}
It is observed that $\tau _{l}^{i}$ monotonously decreases with $\beta_i$ while $\tau _{o}^{i}$ monotonously increases with $\beta_i$, as illustrated in Fig.~\ref{optimal_beta}. Thus the total system cost ${\cal O}_{total}^{P_{case1}}$ can be minimized only when $\delta_p^i=\delta_{op}^i$ holds. By solving this equation, $\beta_i^{{case1}^*}$ can be derived, as given in (\ref{optimal_beta_c1}).
\end{proof}

According to Table~\ref{mapping}, the average inference error rate of $U_i$ using partial offloading scheme is achieved as 
\begin{equation}\label{partial_inference_error}
\epsilon_i = \left(1-\beta_i\right) \tilde{\epsilon_L} + \beta_i \tilde{\epsilon_H} = \tilde{\epsilon_L} - \beta_i \left(\tilde{\epsilon_L} - \tilde{\epsilon_H}  \right).
\end{equation}

Substituting (\ref{partial_inference_error}) into (\ref{binary_original_problem}a), $\beta_i\ge \tilde{\beta_i}$ is achieved, where $\tilde{\beta_i} = \cfrac{\tilde{\epsilon_L}-\epsilon_T^i}{\tilde{\epsilon_L}-\tilde{\epsilon_H}}$. Therefore, the optimal offloading ratio in case 1 is obtained as
\begin{equation}\label{optimal_offloading_ratio_c1}
\beta_i^{case1*}={\rm max}\left \{ \hat{\beta}_i^{case1}, \ \tilde{\beta_i} \right \}.
\end{equation}

\subsection{Case 2: Energy-Constrained Objective}
Suppose that UAVs are with serious energy budget and no additional power supply is available. In this case, the energy consumption becomes a main concern, i.e., $\theta=0$. The total system cost can be rewritten as 
\setcounter{equation}{35}
\begin{equation}
{\cal O}_{total}^{P_{case2}} = \sum_{i=1}^{n} \left( \varepsilon_{l}^{i} + \varepsilon_{o}^{i} \right).
\end{equation}
 
 In this case, $\mathscr{P}_P$ can be transformed as $\mathscr{P}_P^2$, which is described as follows.
 
$\mathscr{P}_P^2$ (\textit{Partial - Energy Minimization Problem}):
\begin{equation}
\begin{aligned} \label{energy_problem}
&\underset{\beta_i}{\rm minimize}\;\; \sum_{i=1}^{n} \left( \varepsilon_{l}^{i} + \varepsilon_{o}^{i} \right) \\
& \;\;\;\;\;\;{\rm{s}}{\rm{.t}}{\rm{.}}\;\;\; \mathbf{C1-C3}.
 \end{aligned}
 \end{equation}
  
To analyze the problem $\mathscr{P}_P^2$, we first introduce Lemma 2 as follows.
 
 \begin{lemma}
When $\eta$, $\epsilon_L$ and $\epsilon_H$ are given, the optimal offloading ratio for $\mathscr{P}_P^2$ can be obtained as
$\beta_i^{case2*}=\tilde{\beta_i}$.
\begin{equation} \label{optimal_beta_case2}
\beta_i^{case2*}=\left\{\begin{matrix}
1, & {\rm If}\  \gamma_i < \gamma_T^{case2}, \\ 
 \tilde{\beta_i}, &  {\rm If} \ \gamma_i > \gamma_T^{case2}.
\end{matrix}\right.
\end{equation}
\end{lemma}

\begin{proof}
Based on (\ref{partial_local_energy}) and (\ref{partial_energy_offloading}), we can achieve
\begin{equation}
\label{sum_energy_case2}
\varepsilon_{l}^{i} \!+\! \varepsilon_{o}^{i}\!=\! \beta_i \left( \gamma_i \!\left (\!P_t^i  \frac{s_i}{R_i}\!+\! P^i_I \frac{c_i}{f_i} \right ) \!-\! \xi_i \left( \tilde{\epsilon_L} \!-\! \tilde{\epsilon_H} \right) \right) \!+\! \kappa \left ( f_{l}^{i} \right )^2 c_i.
\end{equation}

Let $\gamma_T^{case2}=\cfrac{\xi_i \left( \tilde{\epsilon_L} \!-\! \tilde{\epsilon_H} \right) }{P_t^i  \frac{s_i}{R_i}\!+\! P^i_I \frac{c_i}{f_i}}$, it is observed from (\ref{sum_energy_case2}) that $\varepsilon_{l}^{i} \!+\! \varepsilon_{o}^{i}$ monotonously increases with $\beta_i$ when $\gamma_i > \gamma_T^{case2}$ and monotonously decreases with $\beta_i$ instead. Since 
$\beta_i\ge \tilde{\beta_i}$ and $\epsilon_T^i\in [\tilde{\epsilon_H},\tilde{\epsilon_L}]$ hold, for the former case, we have
$\beta_i^{case2*}={\rm max}\left \{ 0,  \tilde{\beta_i} \right \} = \tilde{\beta_i}$. For the latter case,  we have
$\beta_i^{case2*}={\rm max}\left \{ \tilde{\beta_i} , 1\right \} = 1$. According to the two cases above, $\beta_i^{case2*}$ can be achieved as (\ref{optimal_beta_case2}).
\end{proof}


\subsection{Case 3: Tradeoff between Delay and Energy Consumption}
Without loss of generality, in this case, a middle course is considered taking account of tracking delay and energy consumption,  where $0<\theta<1$. The total system cost in case 3 can be expressed as 
\begin{equation}\label{total_cost_case3}
\begin{aligned}
{\cal O}_{total}^{P_{case3}}  &\!=\!\sum_{i=1}^{n} \left( \theta \left(  \frac{c_i}{f_l^i} + {\rm max} \{\delta_p^i,  \delta_{op}^i  \} \right) + \left ( 1-\theta \right ) \left ( \varepsilon_{l}^{i}+\varepsilon_{o}^{i} \right ) \right)\\
&\!=\!\begin{cases}
\sum_{i=1}^{n} \!\left(\! \theta \left( \frac{c_i}{f_l^i} + \delta_p^i \right) \!+\! \left ( 1\!-\!\theta \right ) \left ( \varepsilon_{l}^{i}\!+\!\varepsilon_{o}^{i} \right ) \!\right)\!, \text{If } \delta_p^i \geq \delta_{op}^i,  \\ 
\sum_{i=1}^{n} \!\left(\! \theta \left( \frac{c_i}{f_l^i} + \delta_{op}^i \right) \!+\! \left ( 1\!-\!\theta \right ) \left ( \varepsilon_{l}^{i}\!+\!\varepsilon_{o}^{i} \right ) \!\right)\!, \text{If } \delta_p^i \leq \delta_{op}^i.
\end{cases}
\end{aligned}
\end{equation}

Therefore, $\mathscr{P}_P$ can be transformed as $\mathscr{P}_P^3$, as described below.
 
$\mathscr{P}_P^3$ (\textit{Partial - Weighted-sum Cost Minimization Problem}):
\begin{equation}
\begin{aligned} \label{cost_problem}
&\underset{\beta_i}{\rm minimize}\;\; {\cal O}_{total}^{P_{case3}} \\
& \;\;\;\;\;\;{\rm{s}}{\rm{.t}}{\rm{.}}\;\;\; \mathbf{C1-C3}.
 \end{aligned}
 \end{equation}
 
 Next, we analyze the optimal offloading ratio in the following two conditions: 1) $\delta_p^i \geq \delta_{op}^i$, and 2) $\delta_p^i \leq  \delta_{op}^i$.
 
 1) $\delta_p^i \geq \delta_{op}^i$.

We let $\gamma_i^{T1} \!=\! \cfrac{\frac{\theta \rho_i \tilde{\epsilon_L}}{1-\theta}+\xi_i \left( \tilde{\epsilon_L} \!-\! \tilde{\epsilon_H} \right)}{P_t^i  \frac{s_i}{R_i}\!+\! P^i_I \frac{c_i}{f_i}}, \notag
$
to resolve $\mathscr{P}_P^3$ when $\delta_p^i \geq \delta_{op}^i$, the following Lemma 3 is given.

\begin{lemma}
The optimal offloading ratio in case 3 when $\delta_p^i \geq \delta_{op}^i$ is achieved as
\begin{equation}
\label{Optimal_beta_case3}
\beta_i^{case3*}\!=\! \left\{\begin{matrix}
{\rm min}\left \{ \hat{\beta}_i^{case1}, \ \tilde{\beta_i} \right \}, & {\rm If \ } \gamma_i > \gamma_i^{T1}, \\ 
\hat{\beta}_i^{case1}, & {\rm If \ } \gamma_i < \gamma_i^{T1}.
\end{matrix}\right.
\end{equation}
\end{lemma}

\begin{proof}
See Appendix A.
\end{proof}

\newtheorem{theorem}{\textbf{Theorem}}

 2) $\delta_p^i \leq \delta_{op}^i$. 
 
 Similarly, let 
$
\gamma_i^{T2} \!=\! \cfrac{\left(1-\theta \right)\xi_i \left( \tilde{\epsilon_L} \!-\! \tilde{\epsilon_H} \right)-\theta \rho_i \tilde{\epsilon_H}}{\theta \left(  \frac{s_i}{R_i}\!+\! \frac{c_i}{f_i} \right) + \left(1-\theta \right) \left( P_t^i  \frac{s_i}{R_i}\!+\! P^i_I \frac{c_i}{f_i} \right)}, \notag
$
 to resolve $\mathscr{P}_P^3$ when $\delta_p^i \leq \delta_{op}^i$, the following Lemma 4 is given.

\begin{lemma}
The optimal offloading ratio in case 3 when $\delta_p^i \leq \delta_{op}^i$ is achieved as
\begin{equation}
\label{Optimal_beta_case3}
\beta_i^{case3*}\!=\! \left\{\begin{matrix}
{\rm max}\left \{ \hat{\beta}_i^{case1}, \ \tilde{\beta_i} \right \}, & {\rm If \ } \gamma_i > \gamma_i^{T2}, \\ 
1, & {\rm If \ } \gamma_i < \gamma_i^{T2}.
\end{matrix}\right.
\end{equation}
\end{lemma}

\begin{proof}
The proof is similar to the proof in Lemma 3 and omitted to save space.
\end{proof}

In order to make the optimal solutions given in Lemma 3 and Lemma 4 more clear and easier to follow, the following Corollary 1 is presented.
\newtheorem{corollary}{\textit{Corollary}}
\begin{corollary}
Once $\theta$, $\rho_i$, $\tilde{\epsilon_L}$ and $\tilde{\epsilon_H}$ are given, then $\gamma_i^{T1} > \gamma_i^{T2}$ is obtained.
According to Lemma 3 and Lemma 4, the optimal offloading ratio in case 3 can be concluded as
\begin{equation}
\label{Optimal_beta_case3}
\beta_i^{case3*}\!=\! \left\{\begin{matrix}
1, & {\rm If \ } \gamma_i < \gamma_i^{T2}, \\ 
{\rm max}\left \{ \hat{\beta}_i^{case1}, \ \tilde{\beta_i} \right \}, & {\rm If \ } \gamma_i^{T2}<\gamma_i < \gamma_i^{T1}, \\ 
{\rm min}\left \{ \hat{\beta}_i^{case1}, \ \tilde{\beta_i} \right \}, & {\rm If \ } \gamma_i > \gamma_i^{T1}.
\end{matrix}\right.
\end{equation}
\end{corollary}

\begin{proof}
Dividing the numerator and denominator of $\gamma_i^{T2}$ by $1-\theta$, we can achieve  
\begin{equation}
\gamma_i^{T2} \!=\! \cfrac{\xi_i \left( \tilde{\epsilon_L} \!-\! \tilde{\epsilon_H} \right)-\frac{\theta \rho_i \tilde{\epsilon_H}}{\left(1-\theta \right)}}{\frac{\theta}{\left(1-\theta \right)} \left(  \frac{s_i}{R_i}\!+\! \frac{c_i}{f_i} \right) + \left( P_t^i  \frac{s_i}{R_i}\!+\! P^i_I \frac{c_i}{f_i} \right)}. \notag
\end{equation}

By comparing $\gamma_i^{T2}$ with $\gamma_i^{T1}$, it is observed that $\gamma_i^{T2}<\gamma_i^{T1}$ holds. Therefore, the value range of $\gamma_i$ can be divided into three sections, as indicated in (\ref{Optimal_beta_case3}).
\end{proof}

\begin{remark}
It is observed that $\tilde{\beta_i} \le \beta_i\le \hat{\beta}_i^{case1}$ is obtained in Lemma 3 when $\delta_p^i \geq \delta_{op}^i$ and $A>0$. Note that the derived $\beta_i^{case3*}$ is constraint but global optimal. Once $\beta_i^{case3*}$ cannot meet the criteria $\mathbf{C1}$ (e.g., when $\tilde{\beta_i} > \hat{\beta}_i^{case1}$), this means that there exists no optimal solution. Suppose that the optimal offloading ratio exists leading to   $\tilde{\beta_i} < \hat{\beta}_i^{case1}$, then (\ref{Optimal_beta_case3}) can be further rewritten as
\begin{equation}
\label{Optimal_beta_case3_new}
\beta_i^{case3*}\!=\! \left\{\begin{matrix}
1, & {\rm If \ } \gamma_i < \gamma_i^{T2}, \ (\rm Con. A) \\ 
\hat{\beta}_i^{case1}, & {\rm If \ } \gamma_i^{T2}<\gamma_i < \gamma_i^{T1}, \ (\rm Con. B) \\ 
\tilde{\beta_i}, & {\rm If \ } \gamma_i > \gamma_i^{T1}, \ (\rm Con. C)
\end{matrix}\right..
\end{equation}
\end{remark}


\subsection{A Special Case}
 To gain more insights into the proposed DL tasks offloading framework, we further investigate a specific scenario where \textit{the UAVs are able to distribute the DL tasks according to the input data quality}. Suppose that the memory of the UAV is adequate so as to be able to cache enough data frames. With the help of preprocessing at the UAV, the video frames can be roughly divided into two categories: ``Good" or ``Bad".  In order to minimize the inference errors while keeps the delay at a low level, \textit{it is straightforward that the video frames with ``Bad" quality are offloaded to the MES to improve inference accuracy while the ``Good" frames are processed locally to save the bandwidth}.
  As a result, $\eta=0$ holds for the local processing while $\eta=1$ holds for the offloading processing. Therefore, in the special case, we have $\tilde{\epsilon_L}=\epsilon_L$ and $\tilde{\epsilon_H}=\epsilon_H$. Based on (\ref{Optimal_beta_case3_new}), the optimal offloading ratio in the special case ($\beta_i^{s*}$) can be calculated accordingly.
  
\begin{remark}
In this special case, since the proportion of ``Bad" frames captured in the memory of the UAV equals to the DL tasks partition ratio, i.e., $\eta^{s}=\beta_i^{s*}$, it is prone to capture enough frames in the memory to meet the condition above and then offload ``Bad" frames to the MES. 
\end{remark}


\section{Implementation and Numerical Results}
\label{results}
In this section we first present an intuitive implementation example to illustrate the inference delay and inference accuracy of the TX2 and DGX-1, respectively. Then, the numerical results are presented to demonstrate the performance of the proposed HMTD framework and investigate the impact of the critical parameters.


\subsection{Intuitive Implementation Example}
In the intuitive experimental setup example, the camera first captures the video frames, which are cached at the memory to be further processed. A GPU cluster NVIDIA DGX-1\cite{NVIDIA-1} is considered as an MES bearing the higher-level layers of the pre-trained CNN model. The NVIDIA Jetson TX2 \cite{NVIDIA-2} with TensorFlow\cite{tensorflow} can be considered as an UAV in the experimental example, where the lower-level layers are implemented. Note that both of the higher-level layers and lower-level layers belong to a pre-trained CNN model. After the offline training, the CNN model can learn the features of the new input images and help the UAV to detect the target and make the detection region tuning. For the video frames, we use the ImageNet\cite{imagenet} as the dataset due to the similarity between visual target tracking and object detection, where $80 \%$ of the dataset is used for offline training and the remaining $20 \%$ is used for testing. For each inference, we apply $50$ images and measure the processing time. In our testing results, the inference time and inference accuracy of higher layers running on NVIDIA DGX-1 is about $120$ frame per second (fps) and $90 \%$ ($10 \%$ inference error rate), respectively. The inference time and inference accuracy of the lower layers running on the NVIDIA Jetson TX2 is about $20$ fps and $80 \%$ ($20 \%$ inference error rate), respectively. It can be observed that the NVIDIA DGX-1 performs six times as fast as the NVIDIA Jetson TX2. Although the reference values of parameters are obtained from well-known ML benchmarks to guide our simulations, our simulations are not limited by the reference values, and instead cover the entire range of parameters, as detailed below.



 \begin{table}[t]  
  \captionsetup{font={small}} 
\caption{\\ \scshape Critical Parameters and Values } 
\label{Parameters}
\small  
\centering  
\setlength{\abovecaptionskip}{-3pt} 
\setlength{\belowcaptionskip}{-15pt}
\begin{tabular}{|l | l|}  
\hline
\textbf{Parameters} & \textbf{Value} \\
\hline 
 Total number of UAVs ($n$) & $2$ \\
\hline
 DL tasks size ($s$) & $1$ Mbits\\
\hline
UAV's preference coefficient on delay ($\theta$) & $0.5$ \\
 \hline
``Bad" data probability ($\eta$) & $0.5$ \\
 \hline
Data size scale coefficient ($\gamma$) & $7$ \\
\hline 
 CPU frequency of UAV ($f_l$) &  $1$ GHz \\
 \hline 
 CPU frequency of MES ($F$) &  $10$ GHz \\
 \hline
  Inference error rate of lower-level layers ($\epsilon_L$) & $0.3$ \\
  \hline 
  Inference error rate of higher-level layers ($\epsilon_H$) & $0.2$ \\
  \hline 
  Tasks dropping penalty on delay ($\rho$) & $8$ sec \\
  \hline 
Energy efficiency parameter ($\kappa$) & $1\times10^{-28}$ \\%
 \hline  
 UAV's transmission power ($P_t$) & $10$ W \\
  \hline 
 UAV's idle power ($P_I$) & $5$ W \\
 \hline
 White noise power ($\chi^2$) & $7.9\times10^{-13}$ \\
  \hline
\end{tabular}  
\end{table}

\subsection{Simulation Settings}
In the following, numerical results are provided to demonstrate the performance of the proposed HMTD framework. We consider a testing scenario where the UAV tracking a person who is with relatively low moving speed. In our simulation model, suppose that the UAV flies at a fixed height of $100 \ m$ and the horizontal distance between UAV and MES is $20 \ m$, which keeps as a constant in the simulations due to the relatively low speed of the target.
 The channel bandwidth ($B$) is set as $10$ MHz and the channel power gain is set to be $h_0=-50$ dB at the reference distance of $1\ m$\cite{UAV_MEC_2}. Besides, the maximum inference threshold ($\epsilon_T$) is calculated as $\epsilon_T\!=\!\tilde{\epsilon_L}\!-\!e$, where $e$ is set to be $0.1$ in the simulation. Some critical simulation parameters are given in Table~\ref{Parameters}\footnote{If there is no special instruction, the parameters in the simulations shall be set according to this table.}.
 
The attainable performance of the visual target tracking in this paper are characterized both by the total weighted-sum cost that introduced by computing and communication and by the error rate of inferring the target. These performance metrics are evaluated for our proposed optimal offloading strategies (i.e., `binary offloading (BO'), `partial offloading (PO)' and the `PO under special case'), with two benchmark offloading approaches: 1) `Totally Local (TL)', where all the DL tasks are executed locally, e.g., when the wireless channel between UAV and MES is not available, and 2) `Totally Offloading (TO)', denoting that all the DL tasks are offloaded to the MES to improve the inference accuracy. By conducting the simulations, we aim to answer the following questions: 1) Which offloading strategy should be selected for optimal offloading in the visual target tracking with the constraint of the inference error rate? 2) Validation of the optimization of the critical variable (e.g., the offloading ratio $\beta$) to the system performance.  3) What is the impact of different critical parameters (such as $\theta$, $\eta$, $F$, $n$, $\epsilon_T$, etc.) on the total weighted-sum cost and average inference error rate? 

\begin{figure}[t]
\centering
 \captionsetup{font={footnotesize }}
\subfigure[]{
\includegraphics[width=2.2in,height=1.68in]{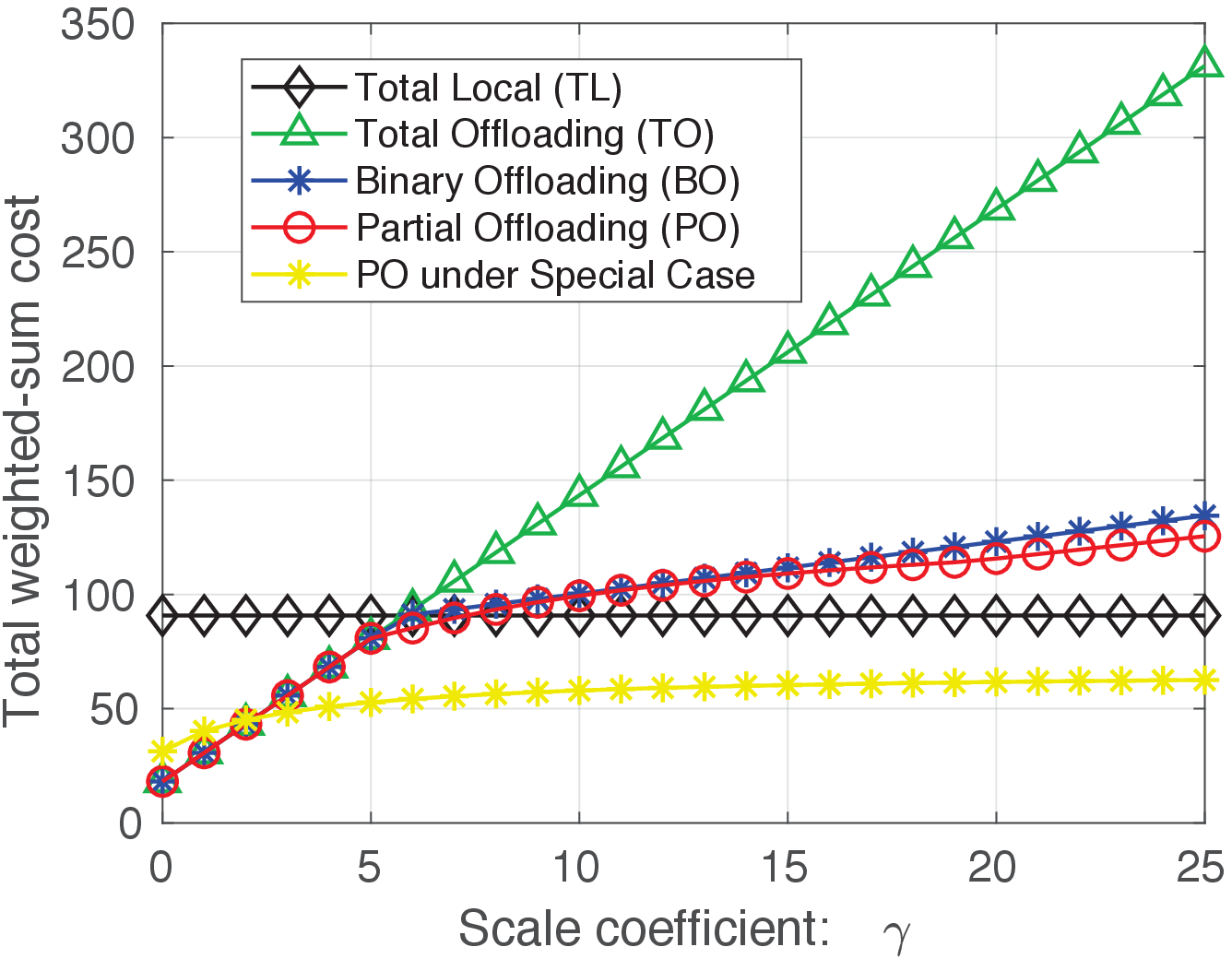}}
\hspace{-0.1in}
\subfigure[]{
\includegraphics[width=2.2in,height=1.68in]{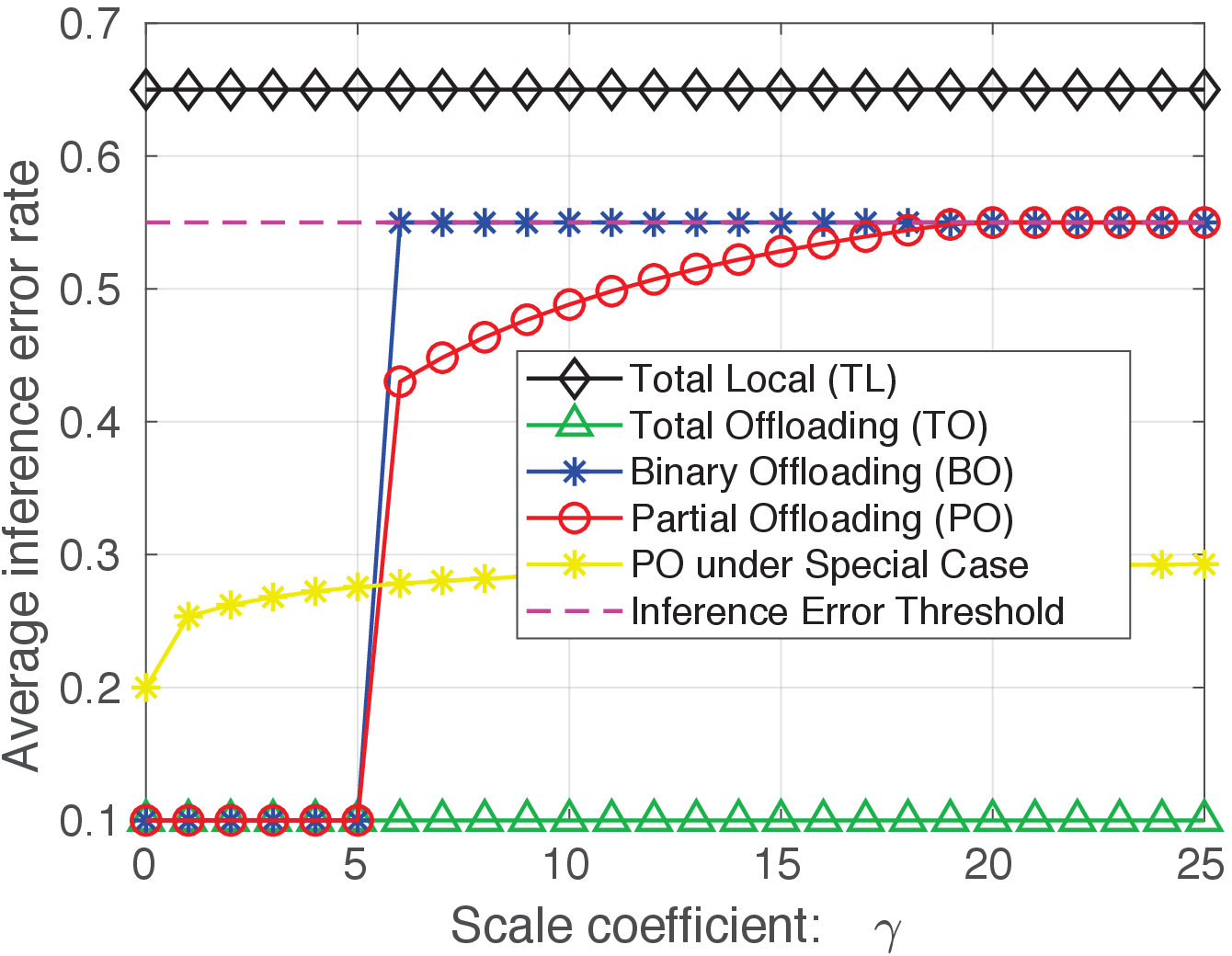}}
\caption{Data size scale coefficient ($\gamma$) versus total weighted-sum cost and average inference error rate are given in (a) and (b), respectively. The results of the proposed optimal BO and PO strategies are obtained from (\ref{binary_optimal_p}) and Corollary 1, respectively. The parameters are listed in Table~\ref{Parameters}, where $N=3$.} 
\label{gamma_impact}
\end{figure}

\subsection{Offloading Option Selection}
One of the four offloading options (i.e. non-offloading, binary offloading, partial offloading and full offloading) is prone to be selected for the visual target tracking based computation offloading scenarios to minimize the total weighted-sum cost, whilst meeting the inference error rate constraint. In this subsection, we aim to investigate the selection of the four offloading strategies under various scale coefficient $\gamma$.

Figs.~\ref{gamma_impact}(a)-(b) depict the total weighted-sum cost and the average inference error rate of the four offloading strategies in the proposed HMTD framework under various values of $\gamma$\footnote{Here $\gamma$ refers to the scale coefficient of the data size output from the higher-level layers at the UAV. The value of $\gamma$ may vary according to the design of the deep learning model. Taking CNN as an example, the value of $\gamma$ mainly depends on the number of filters. That is to say, $\gamma < 1$ can be achieved when pooling layers are used while $\gamma \gg 1$ may hold due to the deployment of filters.}, as shown in Fig.~\ref{gamma_impact}(a) and Fig.~\ref{gamma_impact}(b), respectively.
It can be observed that TO, BO and PO are prone to be selected when $\gamma$ is of a small value (e.g., $\gamma \le 5$ in the figure), which corresponds to the optimal condition (i.e., Con. A) derived in (\ref{binary_optimal_p}) and (\ref{Optimal_beta_case3_new}) in Section~\ref{binary_offloading_framework} and Section~\ref{partial_offloading_framwork}, respectively. When $\gamma$ increases (e.g., $6 \le \gamma \le 18$ in the figure), it is observed that both of the cost and inference error become larger with totally offloading the DL tasks to the MES. This is because the size of the intermediate data output from the higher-level layers becomes larger, which introducing a larger wireless transmission delay. In this situation, it is a better way to process some of the tasks locally and offload the remaining part to the MES (i.e., partial offloading strategy, corresponding to the optimal condition, Con. B derived in (\ref{Optimal_beta_case3_new})), until reaching the Con. B in (\ref{binary_optimal_p}) and Con. C in (\ref{Optimal_beta_case3_new}), respectively. Currently, BO and PO achieve similar performance. Moreover, it can be seen that the PO under the special case outperforms other candidate schemes due to the achieved smaller inference error rate, which benefits from the ability to distribute DL tasks based on the data quality. 

\subsection{Validation of Corollary 1}
Figs.~\ref{validation}(a)-(d) validate our analysis on the optimal offloading ratio of the partial offloading derived in Corollary 1, where the data size scale coefficient is respectively set to be $0.7$ and $30$, corresponding to two typical DL model design.
It can be seen that increasing $\beta$ results in a descending trend of total cost in Fig.~\ref{validation}(a) but a rising trend in Fig.~\ref{validation}(c). This is because when $\gamma$ is small, offloading data to the MES can save the energy consumption of UAVs without introducing too much communication delay. Therefore, offloading much data to the MES (i.e., with a larger value of $\beta$) is a good choice to decrease the cost. However, when a higher $\gamma$ is invoked, the size of intermediate data from the UAVs becomes larger. As a result, much more delay will be introduced by wireless communication, which becomes a bottleneck to the target tracking system. In this situation, processing more data locally (i.e., with a smaller value of $beta$) could be a better choice. Furthermore, it is observed that the optimal offloading ratio is $1$ and $0.2$, with $\gamma=0.7$ and $\gamma=30$, respectively. This observation is consistent with (\ref{Optimal_beta_case3_new}) derived in Corollary 1.

\begin{figure}[t]
\centering
 \captionsetup{font={footnotesize }}
\subfigure[]{
\includegraphics[width=1.65in,height=1.35in]{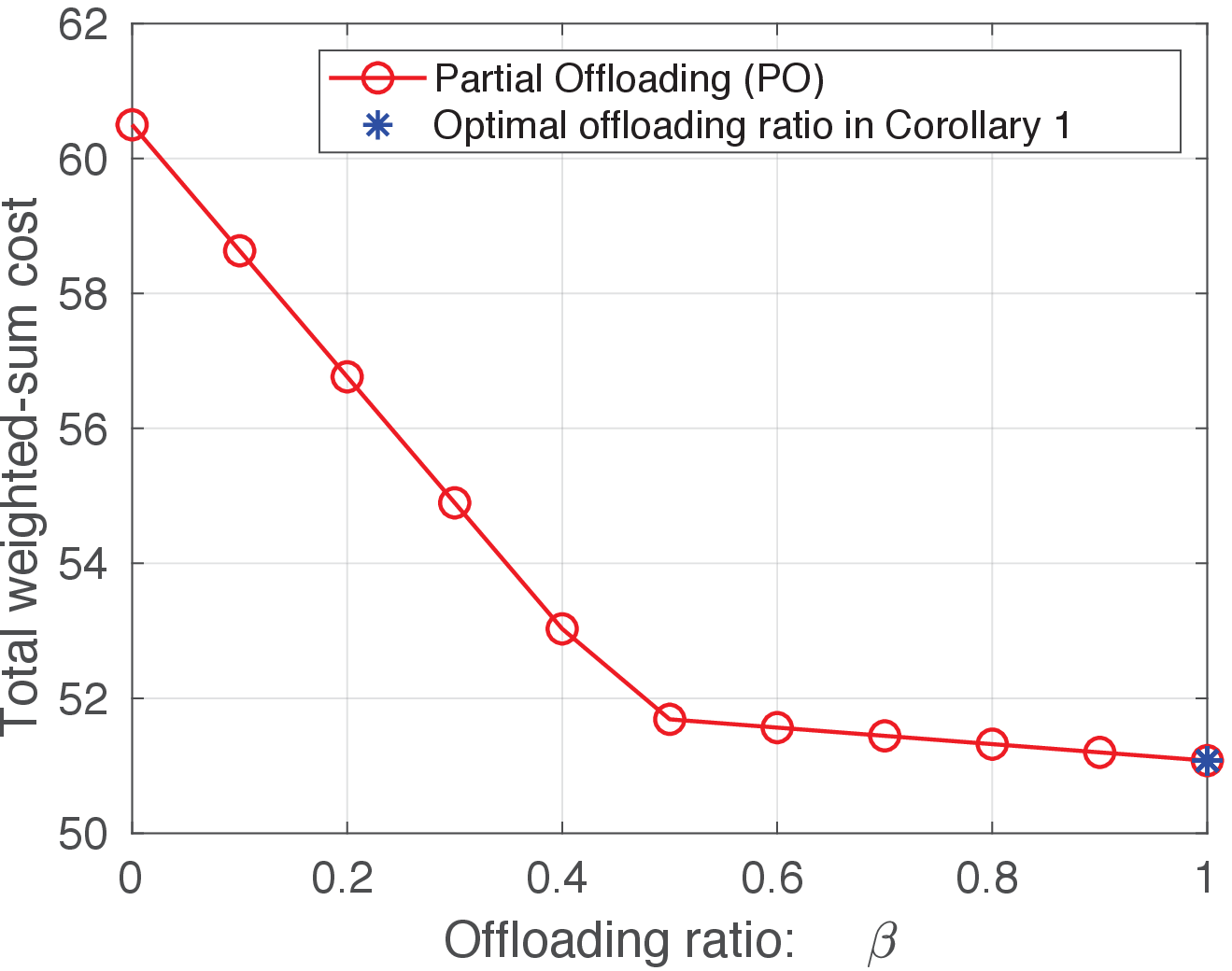}}
\hspace{-0.1in}
\subfigure[]{
\includegraphics[width=1.65in,height=1.35in]{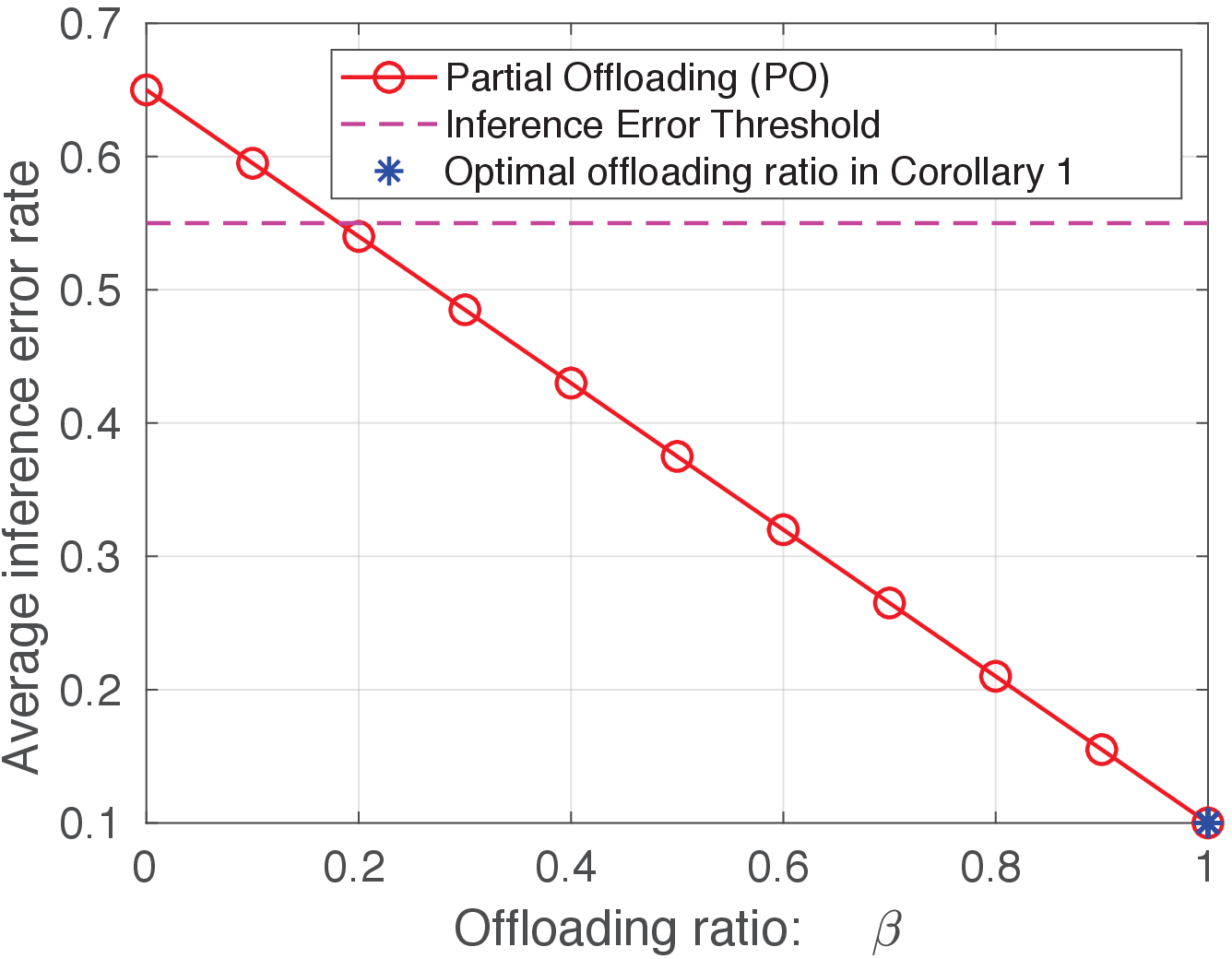}}
\subfigure[]{
\includegraphics[width=1.65in,height=1.35in]{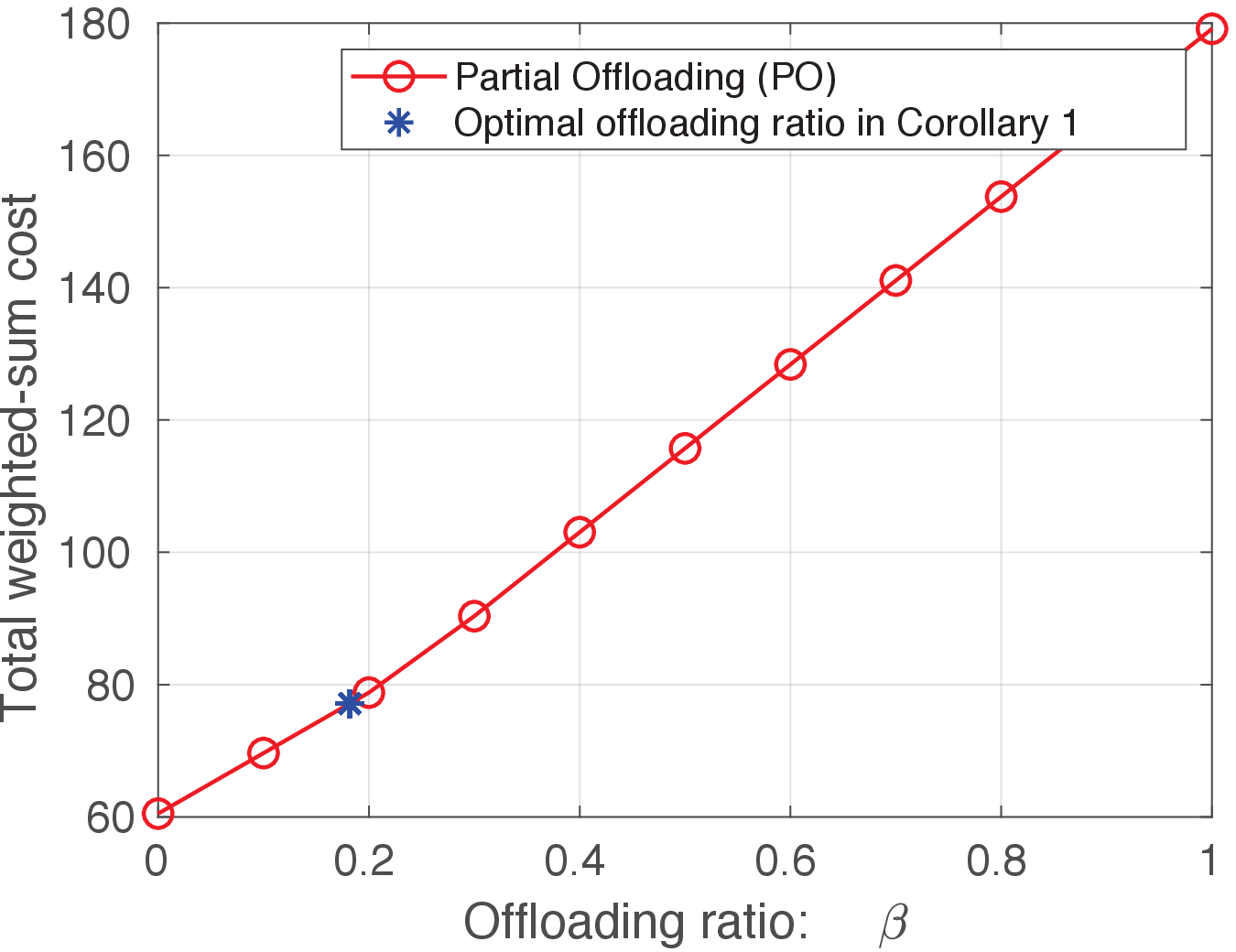}}
\subfigure[]{
\includegraphics[width=1.65in,height=1.35in]{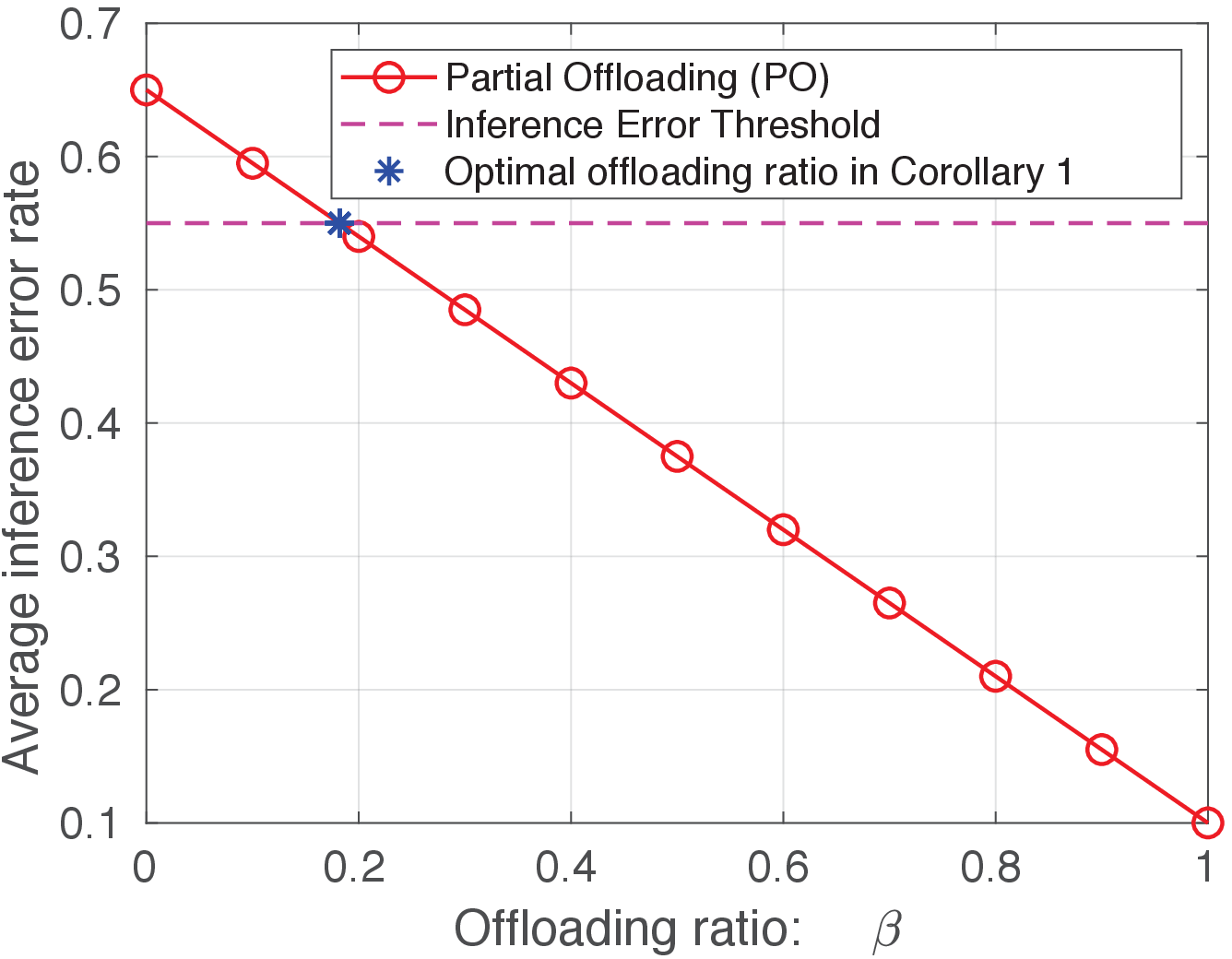}}
\caption{Offloading ratio ($\beta$) versus total weighted-sum cost and average inference error rate with $\gamma=0.7$ in (a)-(b), and with $\gamma=30$ in (c)-(d), respectively. The simulation parameters are listed in Table~\ref{Parameters}.} 
\label{validation}
\end{figure}

\begin{figure}[t]
\centering
 \captionsetup{font={footnotesize }}
\subfigure[]{
\includegraphics[width=1.6in,height=1.35in]{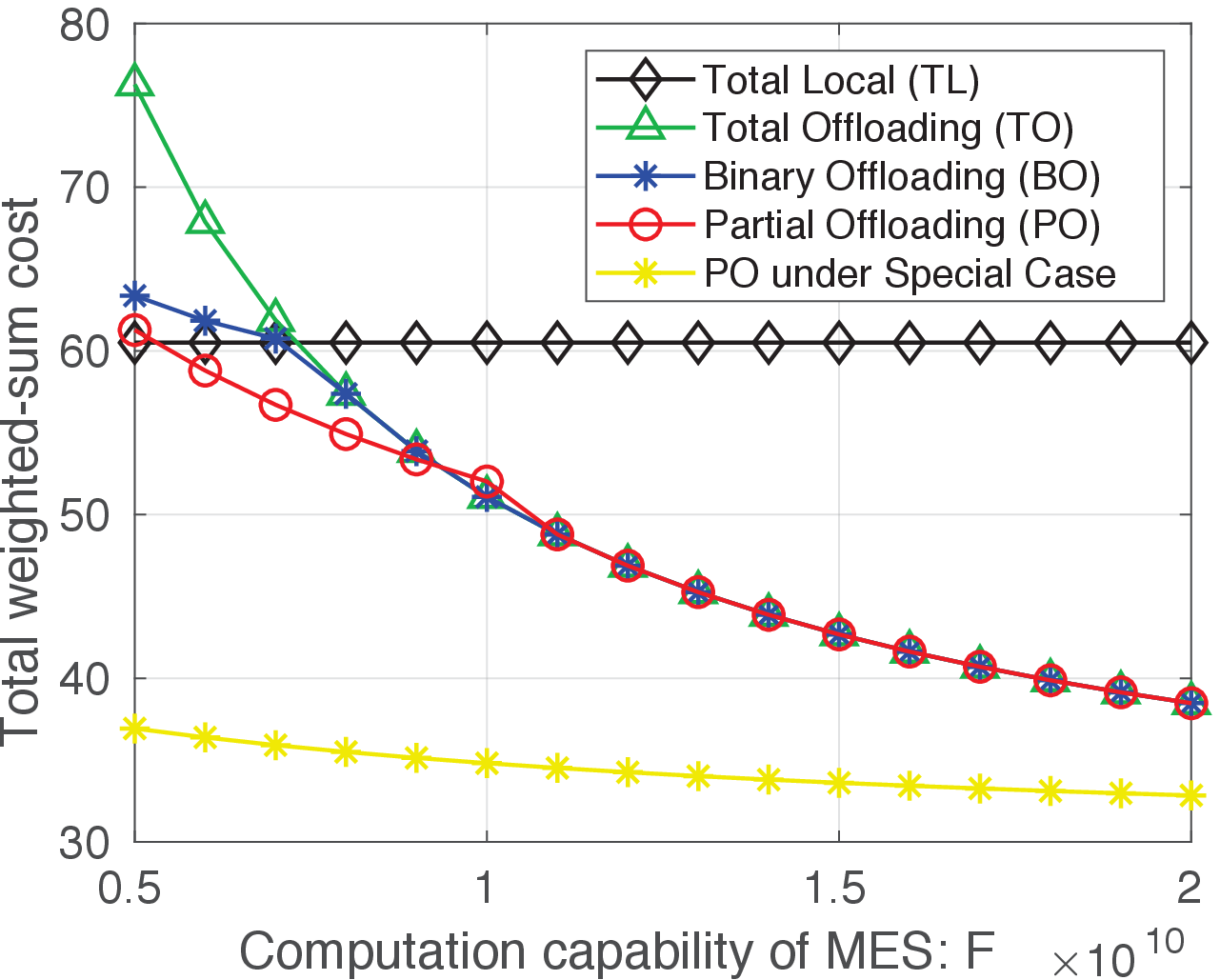}}
\hspace{-0.1in}
\subfigure[]{
\includegraphics[width=1.6in,height=1.35in]{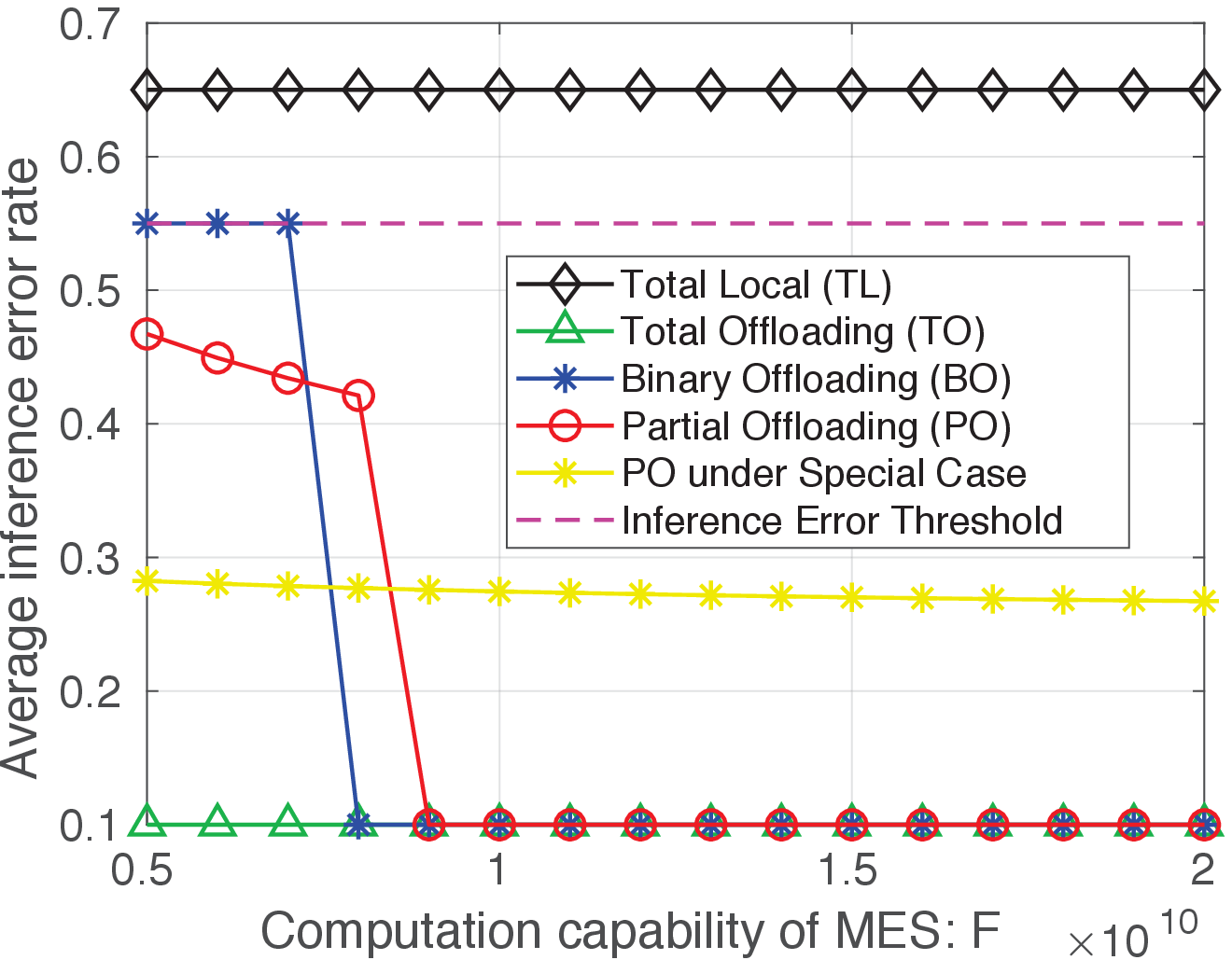}}
\hspace{-0.1in}
\subfigure[]{
\includegraphics[width=1.6in,height=1.35in]{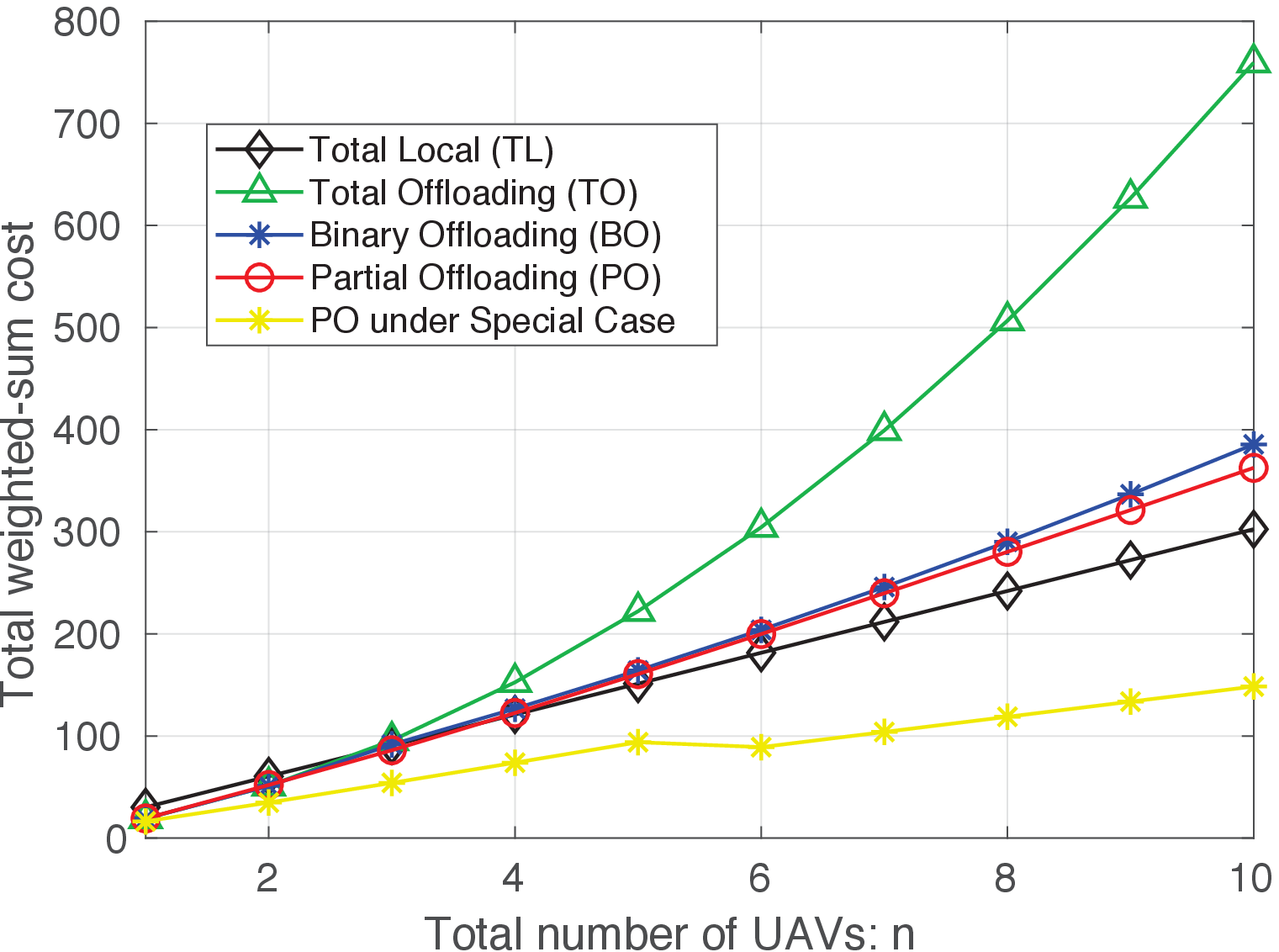}}
\hspace{-0.1in}
\subfigure[]{
\includegraphics[width=1.6in,height=1.35in]{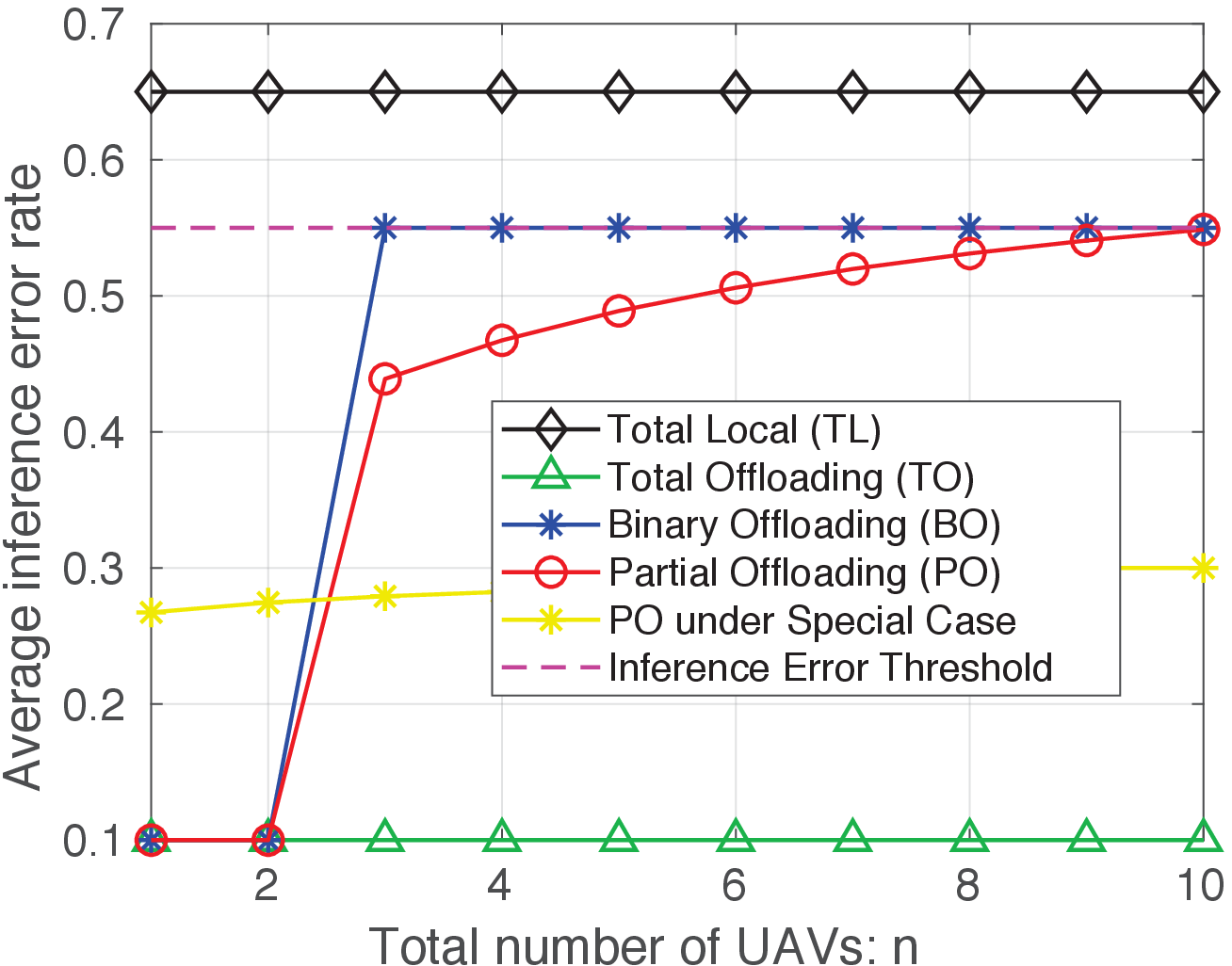}}
\caption{Computation capability of MES ($F$) and total number of UAVs ($n$) versus total weighted-sum cost and average inference error rate are given in (a)-(b) and (c)-(d), respectively. The results of the proposed optimal BO and PO strategies are obtained from (\ref{binary_optimal_p}) and Corollary 1, respectively. The simulation parameters are listed in Table~\ref{Parameters}.}
\label{sim_01}
\end{figure}

\subsection{Impact of the computation capability of MES ($F$) and of the total number of UAVs ($n$)}
Figs.~\ref{sim_01}(c)-(f) depict the impact of $F$ and $n$ on the total weighted-sum cost and average inference error rate, respectively.
As $F$ increases, which indicates that the computation capability of the MES is improved, we can observe from Figs.~\ref{sim_01}(c)-(d) that except for the TL scheme, the cost and inference error rate of the other four schemes is gradually decreased. When $F$ becomes large enough, e.g., $F \ge 10$ GHz, the BO and PO strategies switch into the totally offloading mode (i.e., $\mu^*=\beta^*=1$ is selected), which outperforms the TL significantly. 
Figs.~\ref{sim_01}(e)-(f) evaluate the total weighted-sum cost and average inference error rate for various values of the total number of UAVs ($n$).
The observations are illustrated as follows. Firstly, the advantage of offloading is granted when $n$ is of a low value, e.g., $n \le 4$, by the fact that the cost introduced by wireless communication between UAVs and MES is not too large. Secondly, the increase of $n$ is capable of drastically increasing the total cost of TO when $n$ becomes larger, whereas the inference error rate of TO always keeps at the smallest. Thirdly, PO outperforms BO within the $n$ range of $3$ to $8$, especially the inference error rate, which can be explained with the aid of fine-grained selection of $\beta$ in Corollary 1. Explicitly, it is observed in Fig.~\ref{sim_01}(f) that when $n$ reaches $3$, the achieved inference error rate of PO gradually approaches to the inference error rate threshold $\epsilon_T$, whereas the BO has reached at $\epsilon_T$ already. Furthermore, The advantage of
PO over BO is marginal when we set L$n$ as a large value since the allocated wireless bandwidth becomes less upon increasing $n$.

\begin{figure}[t]
\centering
 \captionsetup{font={footnotesize }}
\subfigure[]{
\includegraphics[width=1.6in,height=1.35in]{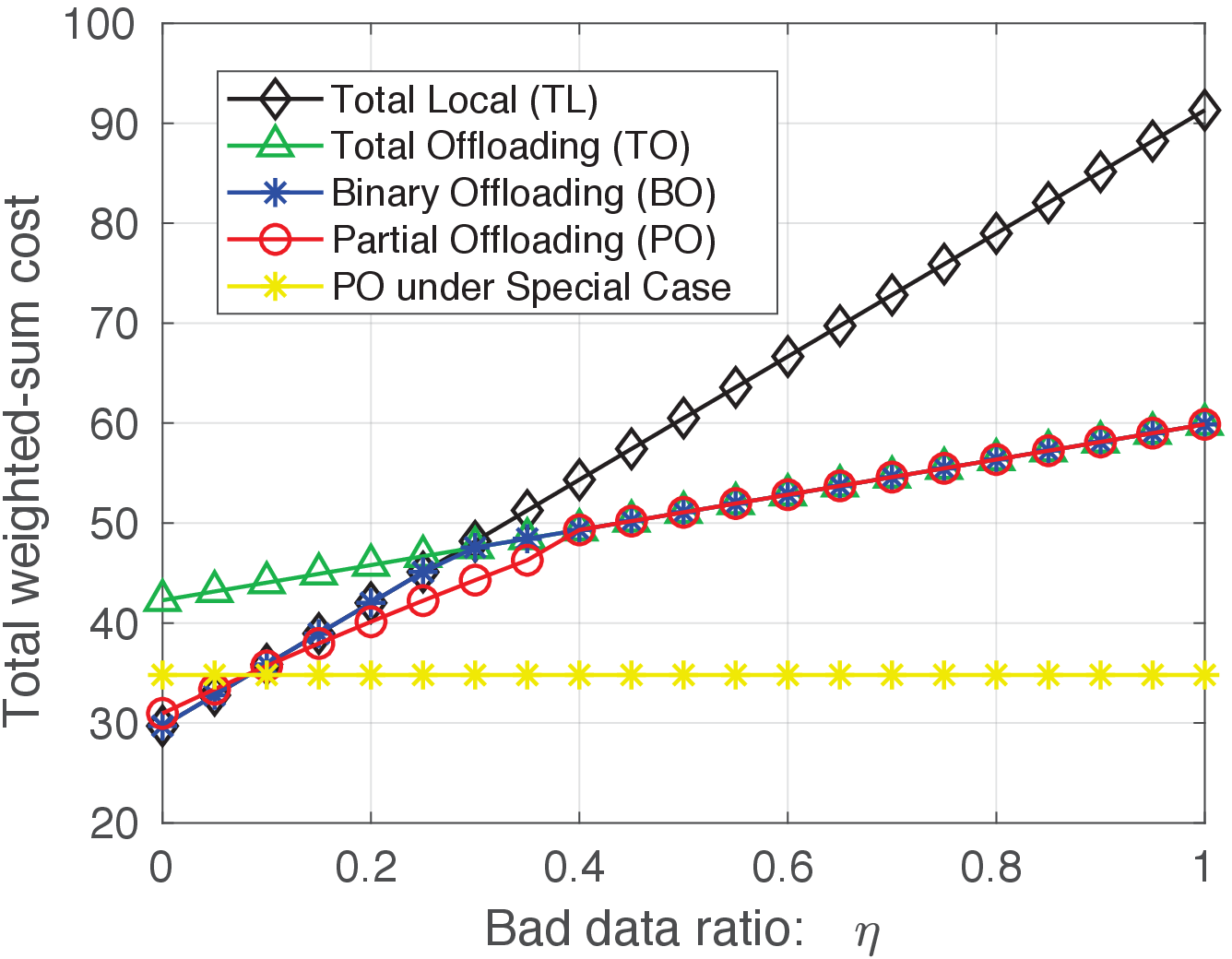}}
\hspace{-0.1in}
\subfigure[]{
\includegraphics[width=1.6in,height=1.35in]{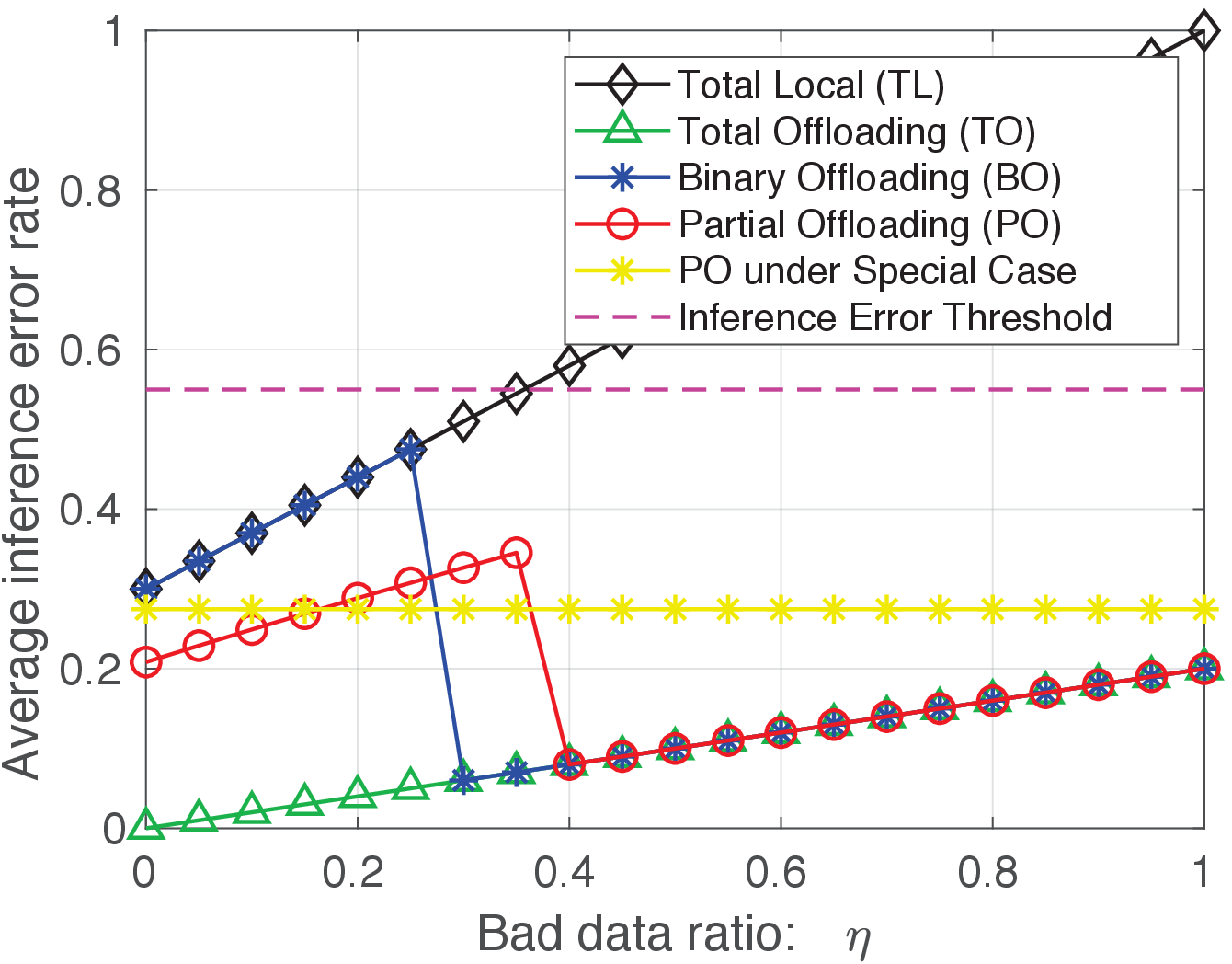}}
\subfigure[]{
\includegraphics[width=1.6in,height=1.35in]{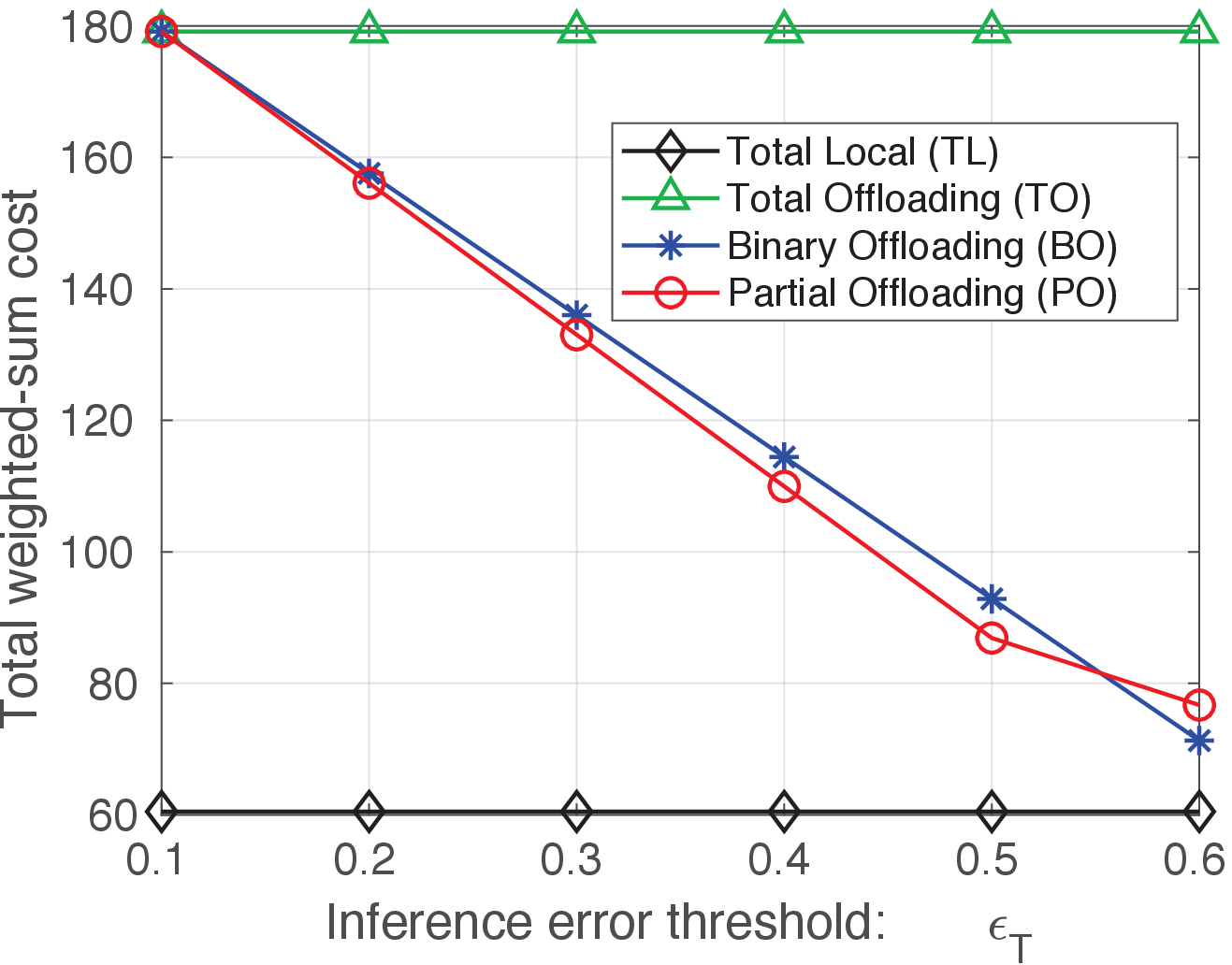}}
\hspace{-0.1in}
\subfigure[]{
\includegraphics[width=1.6in,height=1.35in]{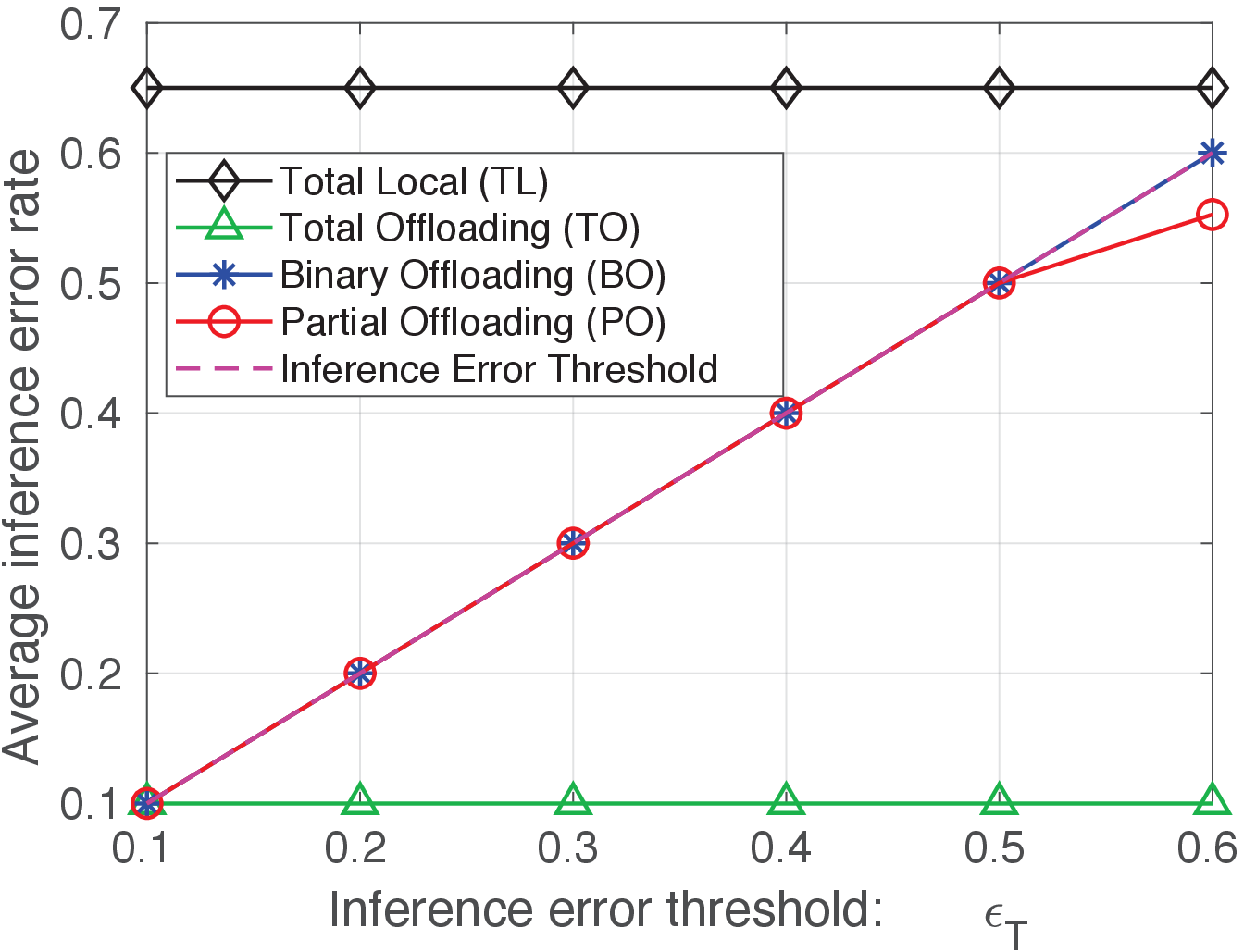}}
\caption{Bad data ratio ($\eta$) and inference error rate threshold ($\epsilon_T$) versus total weighted-sum cost and average inference error rate are given in (a)-(b) and (c)-(d), respectively. The results of the proposed optimal BO and PO strategies are obtained from (\ref{binary_optimal_p}) and Corollary 1, respectively. The simulation parameters are listed in Table~\ref{Parameters}, where we set $\gamma=30$ in (c) and (d).}
\label{sim_02}
\end{figure}

\subsection{Impact of the bad data ratio ($\eta)$ and of the inference error threshold ($\epsilon_T$)}
Figs.~\ref{sim_02}(a)-(b) characterize the impact of bad data rate on the weighted-sum cost and inference error rate using different offloading strategies.
It can be observed that increasing $\eta$ results in both increased cost and inference error rate of all the offloading strategies except the ``PO under special case", as illustrated in Fig.~\ref{sim_02}(a) and Fig.~\ref{sim_02}(b). When $\eta$ is small, e.g., $\eta<0.4$, the local computing is more competitive due to the low achieved cost but high inference accuracy. With the increase of $\eta$, the advantage of TO, BO, and PO schemes is an explicit benefit of exploiting the advantage of offloading when $\eta$ is large, e.g., $\eta \ge 0.4$. Furthermore, we can see that ``PO under special case" always keeps unchanged, which is because that the inference error rate of UAV and MES do not change with $\eta$ leading to a static optimal offloading ratio.

Let us now depict the achieved weighted-sum cost and inference error rate for different values of the inference error threshold ($\epsilon_T$) in Figs.~\ref{sim_02}(c)-(d). It can be observed
that the increase of $\epsilon_T$ results in the reduction of total cost and in the increase of the inference error rate for both BO and PO strategies. This is because the increase of $\epsilon_T$ can lead to the reduction of the necessity of offloading, hence further reduces the total cost, whilst satisfying the inference error rate condition, as verified in Fig.~\ref{sim_02}(d). Moreover, we can see that both the cost and inference error rate of TL and TO strategies keeps unchanged, which is because the way to process the DL tasks is not related to $\epsilon_T$, i.e., the DL tasks are processed totally at UAVs in TL while the DL tasks are totally offloaded to the MES in TO.

\section{Conclusion} \label{conclusion}
In this paper, deep learning enabled visual target tracking on UAVs is considered. Due to limited computing resources and tight energy budget of small UAVs, a novel deployment of trained convolutional neural network (CNN) model for target tracking is applied where the lower layers of the CNN are deployed on the UAV, while the corresponding higher layers are deployed at the MEC server. This setup fulfills the need for timely processing of the video images while taking into account the practical constraints. When the image quality is good, the lower layers of the CNN are able to provide enough features for the tracking performance to be acceptable and only local computations on the UAV are carried out. On the other hand, bad image quality would call for further processing through the higher layers of the CNN at the MEC server. In this context, a novel offloading framework is proposed in this paper to address the tradeoff between delay and energy consumption while taking into account many constraints in reality such as varying image quality, communications bandwidth between UAV and the MEC server, as well as resource sharing among multiple UAVs. The derived analytical results allow us to obtain important insights on MEC supported UAV tracking under a realistic environment. Furthermore, most of the analysis and observations will remain valid and they are not subject to the exact formula used in our simulations that for demonstration purpose only.
 
\appendices
\section{Proof of Lemma 3}
According to (\ref{total_cost_case3}), when $\delta_p^i \geq \delta_{op}^i$ holds, we can achieve
\begin{equation} \label{cost_case3_tmp1}
{\cal O}_{total}^{P_{case3}}  \!=\! \sum_{i=1}^{n} \!\left(\! \theta \left( \frac{c_i}{f_l^i} + \delta_p^i \right) \!+\! \left ( 1\!-\!\theta \right ) \left ( \varepsilon_{l}^{i}\!+\!\varepsilon_{o}^{i} \right ) \!\right).
\end{equation}

Substituting $\delta_p^i = \left(1-\beta_i \right) \rho_i \tilde{\epsilon_L}$ into (\ref{cost_case3_tmp1}), we have
\begin{equation} \label{cost_case3_tmp2}
{\cal O}_{total}^{P_{case3}}  \!=\! \sum_{i=1}^{n} \left( \beta_i A + B \right),
\end{equation}
where 
\begin{equation}
A=\left(1-\theta \right) \left( \gamma_i \left (P_t^i  \frac{s_i}{R_i}\!+\! P^i_I \frac{c_i}{f_i} \right ) - \xi_i \left( \tilde{\epsilon_L} \!-\! \tilde{\epsilon_H} \right) \right) - \theta \rho_i \tilde{\epsilon_L}, \notag
\end{equation}
and
\begin{equation}
B=\theta \left(\frac{c_i}{f_i}+\rho_i \tilde{\epsilon_L} \right) + \left(1-\theta \right) \kappa \left ( f_{l}^{i} \right )^2 c_i. \notag
\end{equation}

\begin{itemize}
\item If $A > 0$ holds, i.e., $\gamma_i > \gamma_i^{T1}$ is obtained, then the weighted-sum cost monotonously increases with $\beta_i$. Since $\beta_i\ge \tilde{\beta_i}$ holds meeting the constraint $\mathbf{C1}$ and $\beta_i\le \hat{\beta}_i^{case1}$ is obtained by solving $\delta_p^i \geq \delta_{op}^i$, the optimal offloading ratio is achieved as $\beta_i^{case3*}={\rm min}\left \{ \hat{\beta}_i^{case1}, \ \tilde{\beta_i} \right \}$.
\item  On the contrary, when $A \le 0$, i.e., $\gamma_i < \gamma_i^{T1}$ holds, then the weighted-sum cost monotonously decreases with $\beta_i$. In this case, we can obtain $\beta_i^{case3*}={\rm min}\left \{ \hat{\beta}_i^{case1},  1 \right \}=\hat{\beta}_i^{case1}$.  
\end{itemize}

To this end, we have proved Lemma 3.
%


\ifCLASSOPTIONcaptionsoff
  \newpage
\fi

%

%
%
%




\end{spacing}

\begin{thebibliography}{00}

\bibitem{UAV_application}
Y. Zeng, R. Zhang, and T. J. Lim, ``Wireless communications with unmanned aerial vehicles, Opportunities and challenges," \textit{IEEE Commun. Mag.}, vol. 54, no. 5, pp. 36-42, May 2016.

\bibitem{UAV_application02}
H. Lu, Y. Li, S. Mu, D. Wang, H. Kim, and S. Serikawa, ``Motor anomaly detection for unmanned aerial vehicles using reinforcement learning," \textit{IEEE Internet of Things Journal}, vol. 5, no. 4, pp. 2315-2322, Aug. 2017.

\bibitem{UAV_application03}
M. Wan, G. Gu, W. Qian, K. Ren, X. Maldague, and Q. Chen, ``Unmanned Aerial Vehicle Video-Based Target Tracking Algorithm Using Sparse Representation," \textit{IEEE Internet of Things Journal}, vol. 6, no. 6, pp. 9689-9706, Jul. 2019.

\bibitem{Infocom_YB}
B. Yang, H. H. Wu, X. Cao, X. Li, T.~Kroecker, Z. Han, and L. Qian, ``Intelli-Eye: An UAV Tracking System with Optimized Machine Learning Tasks Offloading," in Proc. \textit{IEEE INFOCOM WKSHPS},  Paris, France, May 2019.

\bibitem{MEC_survey}
Y. Mao, C. You, J. Zhang, K. Huang, and K. B. Letaief, ``A survey on mobile edge computing: The communication perspective," \textit{IEEE Communications Surveys $\&$ Tutorials}, vol. 19, no. 4, pp. 2322-2358, Aug. 2017.

\bibitem{UAV_AI01}
A. Singh, D. Patil, and S. N. Omkar, ``Eye in the Sky: Real-time Drone Surveillance System (DSS) for Violent Individuals Identification using ScatterNet Hybrid Deep Learning Network," in Proc. \textit{IEEE Conference on Computer Vision and Pattern Recognition (CVPR) Workshops}, Salt lake city, USA, Jun. 2018.

\bibitem{UAV_AI02}
W. G. La, H. Kim, ``Drone Detection and Identification System using Artificial Intelligence," in Proc. \textit{IEEE Information and Communication Technology Convergence (ICTC)}, South Korea, Oct. 2018.

\bibitem{survey-1}
M. Chen, U. Challita, W. Saad, C. Yin, and M. Debbah, ``Machine learning for wireless networks with artificial intelligence: A tutorial on neural networks," \textit{IEEE Communications Surveys $\&$ Tutorials}, vol. 21, no. 4, Jul. 2019. 

\bibitem{vtc}
B. Yang, X. Cao, and L. Qian, ``A Scalable MAC Framework for Internet of Things Assisted by Machine Learning," in Proc. \textit{IEEE VTC2018-Fall}, Chicago, IL, Aug. 2018.

\bibitem{cvpr}
C. Szegedy, W. Liu, Y. Jia, et al., ``Going deeper with convolutions," in Proc. \textit{IEEE conference on computer vision and pattern recognition (CVPR)}, Boston, MA, Jun. 2015. 

\bibitem{Wu}
H. H. Wu, Z. Zhou, M. Feng, Y. Yan, H. Xu,  and L. Qian, ``Real-time single object detection on the uav," In Proc. of \textit{IEEE International Conference on Unmanned Aircraft Systems (ICUAS)}, Atlanta, GA, USA, Jun. 2019.

\bibitem{DJDrone}
Spreading Wings S1000, online: https://www.dji.com/spreading-wings-s1000.

\bibitem{TX2}
NVIDIA Jetson TX2 Module, online: https://developer.nvidia.com/embedded/buy/jetson-tx2.

\bibitem{UAV_Surveillance_survey1}
N. H. Motlagh, M. Bagaa, and T. Taleb, `UAV-Based IoT Platform: A Crowd Surveillance Use Case," \textit{IEEE Communications Magazine}, vol. 55, no. 2, pp. 128-134, Feb. 2017.

\bibitem{UAV_Surveillance_survey2}
G. Ding, Q. Wu, L. Zhang, Y. Lin, T. A.  Tsiftsis, and Y. D. Yao, ``An Amateur Drone Surveillance System Based on the Cognitive Internet of Things," \textit{IEEE Communications Magazine}, vol. 56, no. 1, pp. 29-35, Jan. 2018.

\bibitem{UAV_tracking03}
J. Gu, T. Su, Q. Wang, X. Du, and M. Guizani, ``Multiple Moving Targets Surveillance Based on a Cooperative Network for Multi-UAV," \textit{IEEE Communications Magazine}, vol. 56, no. 4, pp. 82-89, Apr. 2018.

 \bibitem{UAV_relay01}
Y. Zeng, R. Zhang, and T. J. Lim, ``Throughput Maximization for UAV-Enabled Mobile Relaying Systems," \textit{IEEE Transactions on Communications}, vol. 64, no. 12, pp. 4983-4996, Dec. 2016.

\bibitem{UAV_relay02}
H. Wang, J. Wang, G. Ding, J. Chen, Y. Li, and Z. Han, ``Spectrum Sharing Planning for Full-Duplex UAV Relaying Systems With Underlaid D2D Communications," \textit{IEEE Journal on Selected Areas in Communications}, vol. 36, no. 9, pp. 1986-1999, Aug. 2018.

\bibitem{UAV_cover00}
S. A. R. Naqvi, S. A. Hassan, H. Pervaiz, and Q. Ni,``Drone-aided communication as a key enabler for 5G and resilient public safety networks," \textit{IEEE Communications Magazine}, vol. 56, no. 1, pp. 36-42, Jan. 2018.

\bibitem{UAV_cover01}
M. Mozaffari, W. Saad, M. Bennis, and M. Debbah, ``Unmanned aerial vehicle with underlaid device-to-device communications: Performance and tradeoffs," \textit{IEEE Trans. Wireless Commun.}, vol. 15, no. 6, pp. 3949-3963, Jun. 2016.

\bibitem{UAV_MEC_3}
F. Cheng, S. Zhang, Z. Li, Y. Chen, N. Zhao,  F. R. Yu, and V. C. Leung, ``UAV trajectory optimization for data offloading at the edge of multiple cells,"\textit{ IEEE Transactions on Vehicular Technology}, vol. 67, no. 7, pp. 6732-6736, Jul. 2018.

\bibitem{UAV_MEC_4}
S. Jeong, O. Simeone, and J. Kang, ``Mobile edge computing via a UAV-mounted cloudlet: Optimization of bit allocation and path planning," \textit{IEEE Transactions on Vehicular Technology}, vol. 67, no. 3, pp. 2049-2063,  Mar. 2018.

\bibitem{UAV_MEC_1}
J. Lyu, Y. Zeng, and R. Zhang, R, ``UAV-aided offloading for cellular hotspot,"  \textit{IEEE Transactions on Wireless Communications}, vol. 17, no. 6, pp. 3988-4001, Jun. 2018.

\bibitem{UAV_MEC_2}
Q. Hu, Y. Cai, G. Yu, Z. Qin, M. Zhao, and G. Y. Li, ``Joint Offloading and Trajectory Design for UAV-Enabled Mobile Edge Computing Systems," \textit{IEEE Internet of Things Journal}, vol. 6, no. 2, pp. 1879-1892, Apr. 2019.

\bibitem{Full-2}
Y. Liu, F. R. Yu, X. Li, H. Ji, and V. C. Leung, ``Hybrid computation offloading in fog and cloud networks with non-orthogonal multiple access," in Proc. \textit{IEEE Conference on Computer Communications Workshops (INFOCOM WKSHPS)}, Honolulu, HI, Apr. 2018.

\bibitem{Full-1}
F. Wang, J. Xu, X. Wang, and S. Cui, ``Joint Offloading and Computing Optimization in Wireless Powered Mobile-Edge Computing System," in Proc. \textit{IEEE ICC},  Paris, France, May 2017.

\bibitem{YB_WCL}
B. Yang, X. Cao, X. Li, C. Yuen, and L. Qian, ``Lessons Learned from Accident of Autonomous Vehicle Testing: An Edge Learning-aided Offloading Framework," \textit{IEEE Wireless Communications Letters}, Mar. 2020.

\bibitem{Partial-1}
J. Ren, G. Yu, Y. Cai, Y. He, and F. Qu, `` Partial Offloading for Latency Minimization in Mobile-Edge Computing," in Proc. \textit{IEEE GLOBECOM}, Singapore, Dec.  2017.

\bibitem{UAV_DL1}
H. Li, K. Ota, and M. Dong, ``Learning IoT in edge: deep learning for the internet of things with edge computing," \textit{IEEE Network}, vol. 32, no. 1, pp. 96-101, Feb. 2018.

\bibitem{YC_TII}
Z. Zhao, R. Zhao, J. Xia, X. Lei, D. Li, C. Yuen, and L. Fan, ``A Novel Framework of Three-Hierarchical Offloading Optimization for MEC in Industrial IoT Networks," \textit{IEEE Transactions on Industrial Informatics}, vol. 16, no. 8, pp. 5424-5434, Aug. 2020.

\bibitem{XK_ITS}
K. Xiong, S. Leng, C. Huang, C. Yuen, and Y. L. Guan, ``Intelligent Task Offloading for Heterogeneous V2X Communications," \textit{IEEE Transactions on Intelligence Transport System}, Jul. 2020, online: arXiv preprint arXiv:2006.15855.

\bibitem{XK_TVT}
K. Xiong, S. Leng, X. Chen, C. Huang, C. Yuen, and Y. L. Guan, ``Communication and Computing Resource Optimization for Connected Autonomous Driving," \textit{IEEE Transactions on Vehicular Technology}, Aug. 2020, online: arXiv preprint arXiv:2006.15875.

\bibitem{error}
O. Bousquet, U. von Luxburg, and G. Rätsch, ``Advanced Lectures on Machine Learning: Revised Lectures," vol. 3176, Springer, 2011.

\bibitem{Harvard}
S. Teerapittayanon, B. McDanel, and H. T. Kung, ``Distributed deep neural networks over the cloud, the edge and end devices," in Proc. \textit{IEEE 37th International Conference on Distributed Computing Systems (ICDCS)}, Atlanta, GA, Jun. 2017.

\bibitem{ICC_YB}
B. Yang, X. Cao, J. Bassey, X. Li, T. Kroecker, L. Qian, ``Computation Offloading in Multi-Access Edge Computing Networks: A Multi-Task Learning Approach," in Proc. \textit{IEEE International Conference on Communications (ICC)},  Shanghai, China, May 2019.

\bibitem{psnr}
{A. Hore, D. Ziou, ``Image quality metrics: PSNR vs. SSIM," In Proc. \textit{IEEE international conference on pattern recognition}, Istanbul, Turkey, Aug. 2010.}

\bibitem{IoU}
H. Rezatofighi, N. Tsoi, J. Gwak, A. Sadeghian, I. Reid, and S. Savarese, ``Generalized intersection over union: A metric and a loss for bounding box regression," In Proc. of \textit{the IEEE Conference on Computer Vision and Pattern Recognition} Long Beach, CA, Jun. 2019.

\bibitem{sensing}
V. V. Shakhov and I. Koo, ``Experiment Design for Parameter Estimation in Probabilistic Sensing Models," \textit{IEEE Sensors Journal}, vol. 17, no. 24, pp. 8431-8437, Dec. 2017.

\bibitem{quality_on_DL01}
S. Dodge, and K. Lina, ``Understanding how image quality affects deep neural networks," In Proc. of \textit{the eighth international conference on quality of multimedia experience (QoMEX)}, Lisbon, Portugal, Jun. 2016.

\bibitem{quality_on_DL02}
V. Sessions, M. Valtorta, ``The Effects of Data Quality on Machine Learning Algorithms," In Proc. of \textit{the 11th International Conference on Information Quality (ICIQ)}, MIT, Cambridge, MA, Nov. 2006.

\bibitem{YB_TWC}
{B. Yang, X. Cao, Z. Han, and L. Qian, ``A machine learning enabled MAC framework for heterogeneous Internet-of-Things networks," \textit{IEEE Transactions on Wireless Communications}, vol. 18, no. 7, pp. 3697-712, May 2019.}



\bibitem{zhang}
J. Hu, H. Zhang, L. Song, Z. Han, and H. V. Poor, ``Reinforcement Learning for a Cellular Internet of UAVs: Protocol Design, Trajectory Control, and Resource Management," \textit{IEEE Wireless Communications,} vol. 27, no. 1, pp. 116-123, Mar. 2020.

\bibitem{Vehicle01}
M. S. Ramanagopal, C. Anderson, R. Vasudevan,  M. Johnson-Roberson, ``Failing to learn: Autonomously identifying perception failures for self-driving cars," \textit{IEEE Robotics and Automation Letters}, vol. 3, no. 4, pp. 3860-3867, Oct. 2018.

\bibitem{k1}
X. Chen, ``Decentralized computation offloading game for mobile cloud computing," \textit{IEEE Transactions on Parallel and Distributed Systems}, vol. 26, no. 4, pp. 974-983, Apr. 2015.

\bibitem{k}
T. D. Burd and R. W. Brodersen, ``Processor design for portable
systems," \textit{Journal of VLSI signal processing systems for signal, image and video technology}, vol. 13, no. 2-3, pp. 203-221, Aug. 1996.

\bibitem{CPU_parallel}
Z. Liang, Y. Liu, T. M. Lok, and K. Huang, ``Multiuser Computation Offloading and Downloading for Edge Computing with Virtualization," \textit{IEEE Transactions on Wireless Communications.}, 	vol. 18, no. 9, pp. 4298-4311, Sep. 2019. 

\bibitem{IoT_YB}
B. Yang, X. Cao, X. Li, Q. Zhang, and L. Qian, ``Mobile Edge Computing based Hierarchical Machine Learning Tasks Distribution for IIoT," \textit{IEEE Internet of Things Journal}, vol. 7, no. 3, pp. 2169-80, Dec. 2019.

\bibitem{charge_UAV01}
W. Chen, S. Zhao, Q. Shi, and R. Zhang, ``Resonant Beam Charging-Powered UAV-Assisted Sensing Data Collection," \textit{IEEE Transactions on Vehicular Technology}, vol. 69, no. 1, pp. 1086-1090, Jan. 2020.

\bibitem{charge_UAV02}
T. J. Nugent and J. T. Kare, ``Laser power for UAVs," Laser Motive White Paper-Power Beaming for UAVs, 2010.

\bibitem{NVIDIA-1}
NVIDIA DGX-1 Essential Instrument of AI Research, online: https://www.nvidia.com/en-us/data-center/dgx-1/.

\bibitem{NVIDIA-2}
NVIDIA Jetson TX2 Module, online: https://developer.nvidia.com/embedded/buy/jetson-tx2.

\bibitem{tensorflow}
M. Abadi, P. Barham, J. Chen, Z. Chen, A. Davis,J. Dean and M. Kudlur, ``Tensorflow: a system for large-scale machine learning," In proc. of \textit{USENIX Symposium on Operating Systems},  Savannah, GA, USA, Nov. 2016.

\bibitem{imagenet}
O. Russakovsky, et al. ``Imagenet large scale visual recognition challenge." \textit{International Journal of Computer Vision}, vol. 115, no. 3, pp. 211-252, Apr. 2015.

















\end{thebibliography}
\end{document}